\documentclass[11pt,a4paper]{article}

\usepackage{afterpage}
\usepackage{algorithm2e}
\usepackage{amsmath}
\usepackage{amssymb}
\usepackage{amsthm}
\usepackage{chngcntr}
\usepackage{xcolor}
\usepackage{dsfont}
\usepackage{graphicx}
\usepackage{setspace}
\usepackage[top=25mm,bottom=35mm,hmargin=25mm]{geometry}
\usepackage{mathtools}
\usepackage{titlesec}
\usepackage[normalem]{ulem}
\usepackage[all]{xy}
\usepackage[utf8]{inputenc}
\usepackage[T1]{fontenc}
\usepackage{lmodern}

\usepackage{hyperref}

\allowdisplaybreaks

\hypersetup{colorlinks=true,urlcolor=blue,linkcolor=black,citecolor=black}
\newcommand*{\doi}[1]{\href{http://dx.doi.org/\detokenize{#1}}{doi}}

\newtheorem{theorem}{Theorem}[section]
\newtheorem{proposition}[equation]{Proposition}

\newtheorem{lemma}[equation]{Lemma}

\theoremstyle{definition}
\newtheorem{remark}[equation]{Remark}
\newtheorem{example}[equation]{Example}

\counterwithin{figure}{section}
\numberwithin{equation}{section}

\newcommand{\N}{\mathbb{N}}

\newcommand{\Z}{\mathbb{Z}}
\newcommand{\cH}{\mathcal{H}}
\newcommand{\T}{\mathcal{T}}
\newcommand{\cX}{\mathcal{X}}
\newcommand{\cW}{\mathcal{W}}
\newcommand{\cF}{\mathcal{F}}
\newcommand{\cM}{\mathcal{M}}
\newcommand{\cZ}{\mathcal{Z}}

\newcommand{\hcX}{\widehat{\cX}}

\newcommand{\hcW}{\widehat{\cW}}
\newcommand{\hT}{\widehat{T}}
\newcommand{\hmu}{\widehat{\mu}}
\newcommand{\htheta}{\widehat{\theta}}
\newcommand{\tmu}{\tilde{\mu}}

\newcommand{\tcZ}{\tilde{\cZ}}
\newcommand{\bxi}{\xi^{\circ}}
\newcommand{\orho}{\bar{\rho}}
\newcommand{\tv}{\tilde{v}}
\newcommand{\tT}{\widetilde T}
\newcommand{\ttheta}{\tilde\theta}

\newcommand{\vep}{\varepsilon}

\newcommand{\ww}{\mathtt r}

\DeclareMathOperator{\Palm}{Palm}

\newcommand{\one}{\mathds{1}}
\newcommand{\dd}{\mathrm{d}}

\renewcommand{\geq}{\geqslant}
\renewcommand{\ge}{\geqslant}
\renewcommand{\leq}{\leqslant}
\renewcommand{\le}{\leqslant}
\renewcommand{\subset}{\subseteq}

\definecolor{midblue}{rgb}{.2,.2,.7}
\definecolor{darkgreen}{rgb}{0,.5,0}
\newcommand{\blue}{\color{blue}}
\newcommand{\red}{\color{red}}
\newcommand{\green}{\color{darkgreen}}

\newcommand{\orange}{\color{orange}}
\newcommand{\purple}{\color{purple}}

\mathtoolsset{showonlyrefs}

\begin{document}

\title{Soliton decomposition of the Box-Ball System}

\author{Pablo A. Ferrari, Chi Nguyen, Leonardo T. Rolla, Minmin Wang}

\maketitle

\begin{abstract}
The Box-Ball System, shortly BBS, was introduced by Takahashi and Satsuma as a discrete counterpart of the KdV equation. Both systems exhibit solitons whose shape and speed are conserved after collision with other solitons. We introduce a slot decomposition of ball configurations, each component being an infinite vector describing the number of size $k$ solitons in each $k$-slot. The dynamics of the components is linear: the $k$-th component moves rigidly at speed $k$. Let $\zeta$ be a translation invariant family of independent random vectors under a summability condition and~$\eta$ the ball configuration with components~$\zeta$. We show that the law of $\eta$ is translation invariant and invariant for the BBS. This recipe allows us to construct a big family of invariant measures, including product measures and stationary Markov chains with ball density less than $\frac{1}{2}$. We also show that starting BBS with an ergodic measure, the position of a tagged $k$-soliton at time $t$, divided by $t$ converges as $t\to\infty$ to an effective speed $v_k$. The vector of speeds satisfies a system of linear equations related with the Generalized Gibbs Ensemble of conservative laws.
\end{abstract}

\begin{figure*}[hb!]
\includegraphics[width=\textwidth]{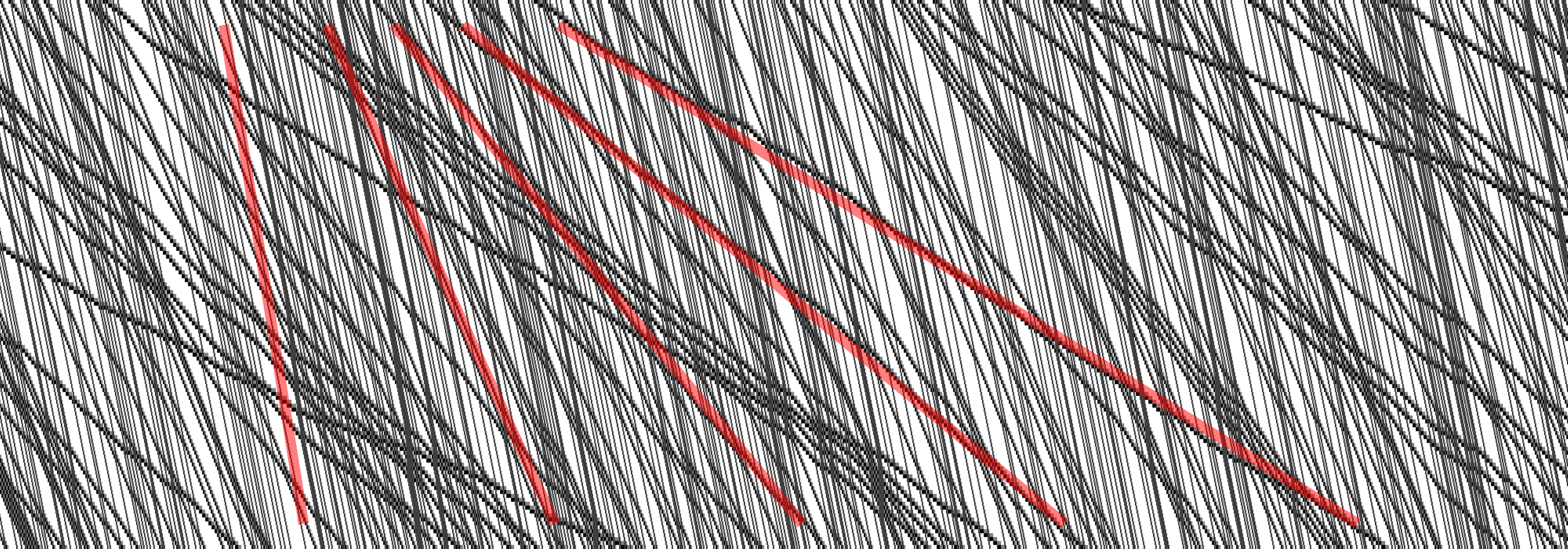}
\\
BBS dynamics for i.i.d.\ initial configuration with density 0.25. Time is going down. Straight red lines are deterministic and computed using Theorem~\ref{thm:speedsexplicit}.
{
\hfill
\small
(high resolution, color online)}
\end{figure*}

\section*{Overview}

Assume that there is a \emph{box} at each integer $x\in\Z$ and that each box may contain a \emph{ball} or be empty. Denote $\eta\in\{0,1\}^\Z$ a ball configuration, with the convention $\eta(x):= 1$ if there is a ball at $x$, else $\eta(x):= 0$.
Imagine a carrier that traverses $ \Z $ from left to right as follows.
When visiting box $x$, the carrier picks a ball if there is one, and deposits a ball if the box $ x $ is empty and the carrier has at least one ball.
Let $T\eta$ be the configuration obtained after the carrier visited all boxes.
An example of $\eta$, carrier load, and $T\eta$ is as follows.
\begin{align*}
\begin{array}{ll}
\hbox{\tt
...0 0 0 1 0 1 1 0 0 0 0 1 1 0 1 0 0 0 0 0} &\eta\\
\hbox{\tt
\ ...0 0 0 1 0 1 2 1 0 0 0 1 2 1 2 1 0 0 0} &\hbox{carrier load}
\\
\hbox{\tt
...0 0 0 0 1 0 0 1 1 0 0 0 0 1 0 1 1 0 0 0 } &T\eta
\end{array}
\end{align*}

The map $T$ for a general $ \eta\in\{0,1\}^\Z $ is defined in~\eqref{t2}.
This dynamics is called \emph{Box-Ball System} (BBS) and was introduced by Takahashi and Satsuma~\cite{TakahashiSatsuma90}, who proposed an algorithm to identify \emph{solitons} in configurations with a finite number of balls and argued that each soliton identified at time 0 can be tracked at successive iterations of $T$.
For example, in the configuration having exactly 3 balls at boxes 1,2,3 the $3$-soliton $\gamma$ consists of these 3 occupied boxes and the empty boxes 4,5,6.
Evolving this configuration by $t$ iterations of $T$, the new configuration will have a $3$-soliton $\gamma^t$ which is a translation of $\gamma$ by $3t$.
In general,~\cite{TakahashiSatsuma90} observed that a $k$-soliton always consists of $k$ occupied boxes and $k$ empty boxes. The relative positions can change and be more scattered during collisions with other solitons but the striking property of the BBS is that such collisions neither create nor destroy solitons. The distance between solitons of the same size is also conserved after collisions.

The main goal of~\cite{TakahashiSatsuma90} was to propose an integrable system with the same behavior as the Korteweg-de Vries equation, KdV, whose solutions include solitons of different sizes that keep shape after collision with other solitons.
Then~\cite{TokihiroTakahashiMatsukidairaSatsuma96,TakahashiMatsukidaira97} argued a way to go from KdV to BBS via {ultradiscretization} and {tropical geometry}, see also~\cite{KatoTsujimotoZuk17,Zuk20}.
Further use of Bethe ansatz to study asymptotic behavior of BBS can be found in~\cite{InoueKunibaOkado04, MadaIdzumiTokihiro06, InoueKunibaTakagi12, KunibaLyuOkado18}. 
The model has also attracted attention in the combinatorics and probability communities.
Using a map between solitons and Dick paths in~\cite{ToriiTakahashiSatsuma96}, the paper~\cite{LevineLyuPike17} relates soliton sizes to longest increasing subsequence of restricted permutations.
\cite{FerrariGabrielli20a} find a new soliton identification that maps to a branch decomposition of the Neveu-Aldous trees of random walks.
See~\cite{LamPylyavskySakamoto21, Sakamoto14, Sakamoto14a} for some other combinatorial developments.
The paper~\cite{HamblyMartinOConnell01} shows that stationary Markov chains are invariant measures for the Pitman transformation~\cite{Pitman75}, a dynamics equivalent to BBS, see~\cite{CroydonKatoSasadaTsujimoto18,CroydonSasada19,croydon2020discrete} for extensions.

This paper has three main contributions.
Firstly, we discover the following fact:
any ball configuration $\eta$ with ball density less than $\frac12$ can be mapped to a family of soliton components, $(\zeta_k)_{k\ge1}$, called the \emph{slot decomposition} of $\eta$; here $\zeta_k(i)$ represents the number of $k$-solitons at coordinate~$i\in\Z$ of the $k$-th component.
The components of a ball configuration evolve linearly under $ T $: component $k$ moves rigidly at speed~$k$. The interaction among components reappears when the ball configuration is reconstructed from the components. This is a delicate hierarchical arrangement where the $k$-th component is \emph{appended} to a certain subset of $\Z$ called $k$-\emph{slots}, determined by the $m$-components for $m>k$. The above facts are purely deterministic. It would be interesting to understand the relationship between the {rigged configurations} in~\cite{InoueKunibaOkado04} with the slot decomposition.

The BBS can be seen as a dynamical system acting on the set of configurations with density less than $ \frac{1}{2} $, and a natural question is about the invariant measures $\mu$ defined by $\mu T^{-1} =\mu$.
Our second main result states that given a translation invariant random family $\zeta$ of independent vectors satisfying a summability condition, the law of the random ball configuration whose slot decomposition is $\zeta$ is translation invariant and invariant for the BBS.
Product measures and stationary Markov chains with density less than $\frac12$ satisfy those properties, as well as a large family of measures based on soliton weights~\cite{FerrariGabrielli20}. We conjecture that the slot decomposition characterizes $T$-invariant probability measures in the sense that, if $\mu$ is shift-mixing and $T$-invariant, then its components should be independent and shift-mixing, see Remark~\ref{rmk:conjecture}.

Our third main contribution is the characterization of the \emph{effective soliton speeds} for shift-ergodic initial states, illustrated by the red lines in the figure of the abstract. The result is based on the following rough description of the soliton dynamics. A size $k$ soliton travels at speed $k$ in absence of other solitons and when two solitons collide, the smaller soliton gets delayed and the bigger soliton ``jumps'' over the slower one. 
Our proof is based on shift-ergodicity and Palm theory, and does not use the slot decomposition. For non-homogeneous initial condition, the effective speed equations have been derived by~\cite{CroydonSasada20}, who also perform a hydrodynamic limit, using our slot decomposition. The results are analogous to those for the hard rod system~\cite{boldrighini1983one}, where disjoint segments of fixed size of the real line called rods travel ballistically at assigned speeds until collision, when the speeds are interchanged. A \emph{pulse} follows the rod that is travelling at one of the given speeds. Hard rod pulses and BBS solitons have a similar dynamics and consequently, similar hydrodynamic and effective speed equations. These results belong to the very active area of generalized hydrodynamics of the generalized Gibbs ensemble~\cite{ spohn2012large,DoyonSpohn17,DoyonYoshimuraCaux18,CaoBulchandaniSpohn19,KunibaMisguichPasquier20} and references therein. 

The paper is organized as follows.
In \S\ref{sec:results} we state the main results of this paper, after giving the definitions needed for the statements.
In \S\ref{sec:solitons} we introduce the slot decomposition of ball configurations, show that the slot decomposition is an injective map (Theorem~\ref{bijection1}), and describe how a configuration can be reconstructed from the components.
In \S\ref{sec:compevol} we show that under the BBS dynamics each component shifts rigidly (Theorems~\ref{simple-linear} and~\ref{thm:hierarchy}).
In \S\ref{sec:measures} we give an explicit construction of $T$-invariant measures that are shift-invariant (Theorem~\ref{thm:invariant}).
In \S\ref{sec:palm} we compile fragments of Palm theory from the literature
that play an important role in many of our arguments.
In \S\ref{sec:speeds} we study the asymptotic speed of tagged solitons (Theorems~\ref{thm:hspeed} and~\ref{thm:speedsexplicit}).
We also study the soliton speeds in terms of tagged records (Theorem~\ref{thm:speeds}).
In \S\ref{sec:postponed} we complete some proofs postponed in previous sections.

\section{Preliminaries and results}
\label{sec:results}

In this section we describe our main results.
We will work mostly with configurations with ball density less than $\frac12$. More precisely, let
\begin{align}
\cX_\lambda :=\Bigl\{\eta\in\{0,1\}^\Z: \lim_{y\to\infty}\,\frac1{y} \sum_{x=-y}^0\eta(x) = \lim_{y\to\infty}\,\frac1{y} \sum_{x=0}^y\eta(x) =\lambda\Bigr\}
\ \
\text{ and }
\ \
\cX := \bigcup_{\mathclap{0 < \lambda < \frac12}} \cX_\lambda
.
\end{align}

In the sequel, a site $x\in\Z$ is often referred to as a \emph{box}.
For $\eta \in \cX$ we define the set of \emph{records} by
\begin{equation}
\label{eq:reta}
R\eta := \Big\{ x\in \Z : \sum_{y=z}^x \eta(y) < \sum_{y=z}^x [1-\eta(y)] \text{ for all } z\le x \Big\}.
\end{equation}
Note that $ \eta(x)=0 $ for all $ x \in R\eta $.
The piece of configuration between two consecutive records forms a finite \emph{excursion}.
The operator $T$ is defined by
\begin{align}
\label{t2}
T\eta(x) :=
\begin{cases}
0 ,& x \in R\eta, \\
1-\eta(x) ,& \text{otherwise}.
\end{cases}
\end{align}
In other words, the value at records stay 0 and the excursions are flipped. When applied to finite ball configurations, this operator coincides with the operator described in the Overview.
We show in \S\ref{sec:solitons} that if $\lambda<\frac12$ then $\cX_\lambda$ is invariant under~$T$.

\subsection{Identifying and tracking solitons}

Define the \emph{runs} of $\eta$ as maximal blocks of successive sites where $\eta$ has a constant value, so that they form a partition of $\Z$.
Assume first that $\eta$ has a finite number of balls, so it has a finite number of finite runs and two semi-infinite runs of zeros, one to the left and one to the right.

A \emph{$k$-soliton} is a collection of $2k$ boxes identified by the Takahashi--Satsuma algorithm~\cite{TakahashiSatsuma90} running as follows.

\begin{algorithm}[H]
Start with a doubly infinite \emph{word}, where each \emph{letter} in the word is $0$ or $1$.
\\
\While{\rm there are still ones in the \emph{word}}{
Select the leftmost run in the \emph{word} whose length is at least as long as the length (denote it $k$) of the run preceding it
\\
Identify a soliton of size $k$, or simply $k$-soliton, consisting of the first $k$ letters of this run and the letters of the run preceding it
\\
Remove these $2k$ letters from the \emph{word}
}
\end{algorithm}

\begin{figure}[b]
\centering
\rule{.95\textwidth}{.5pt}
\par
\bigskip
\includegraphics[width=.9\textwidth]{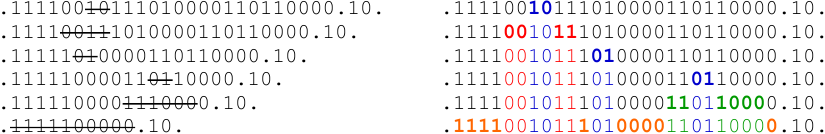}
\caption{\small%
Applying the Takahashi--Satsuma algorithm to a sample configuration. Dots represent records. On the left we have the resulting word after successive iterations. Identified solitons are shown in bold once and then with a color corresponding to their size.
The algorithm is applied to each excursion separately, so the rightmost $1$-soliton in the picture is ignored by this instance of the procedure.
(color online)}
\label{fig:algo}
\end{figure}

Notice that a $k$-soliton consists of $k$ zeros (possibly non-consecutive) followed by $k$ ones or vice-versa, and letters which do not belong to any soliton are all zero and correspond to the records of $\eta$, see Fig.~\ref{fig:algo}.
For a general $\eta \in \cX$, all excursions have a finite number of boxes.
To identify the solitons in $\eta$, we take each excursion of $\eta$, append infinitely many zeros to the left and right of the excursion, and then apply the above algorithm to it.

We define the \emph{head} and \emph{tail} of a $k$-soliton $\gamma$ as follows: the head $\cH(\gamma)=\{\cH_{1}, \cH_{2}, \dots, \cH_{k}\}\subset \Z$ is the set of $k$ boxes with ones in $\gamma$ and the tail $\T(\gamma)=\{\T_{1}, \T_{2}, \dots, \T_{k}\}\subset \Z$ is the set of $k$ boxes with zeroes in $\gamma$.
Namely, $\cH(\gamma)\cup \T(\gamma)$ is the set of boxes that are removed when executing the previous algorithm on $\eta$.
Let $\Gamma_k\eta$ be the set of $k$-solitons of a ball configuration $\eta\in \cX$.
The following is proved in \S\ref{sec:postponed}.

\begin{proposition}
\label{prop:solitontrack}
For any $\eta\in\cX$ and $A \subseteq \Z$, there is a $k$-soliton $\gamma\in\Gamma_k\eta$ with tail $\T(\gamma)=A$ if and only if there is a $k$-soliton $\gamma^1 \in \Gamma_k(T\eta)$ with head $\cH(\gamma^1) = A$.
\end{proposition}

By the above proposition, we can track each $k$-soliton $\gamma$ in the evolution of $\eta$. For each $k$-soliton $\gamma \in \Gamma_k \eta$, call $(\gamma^t)_{t\ge 0}$ the trajectory satisfying $\gamma^0 =\gamma$, $\gamma^t\in \Gamma_kT^t\eta$
and
\begin{equation}
\label{eq:trackgamma}
\cH(\gamma^{t+1}) = \T(\gamma^t)
.
\end{equation}

\subsection{Soliton nesting and motion}

As shown in Fig.~\ref{fig:algo}, solitons can be nested inside larger solitons.
As it turns out, they are nested in a hierarchical way.
Moreover, solitons only move to the right, and they are only free to move when they are not nested inside larger solitons.
In particular, solitons can only overtake smaller solitons.
We now make these statements precise.

Let $ \gamma \in \Gamma_k \eta $ for some $ k\in\N $.
Let us denote by $x(\gamma)$ its leftmost site.
Take $ z $ as the first site to the right of $ \gamma $ such that $ z $ is either a record or belongs to another $ m $-soliton for some $ m \geq k $.
The \emph{interval spanned by $ \gamma $} is defined as $ I(\gamma) := [x(\gamma),z-1] \cap \Z $.

\begin{lemma}
\label{lem:basics}
Let $ \gamma \in \Gamma_k \eta $.
Then both the head and tail of $ \gamma $ are contained in $ I(\gamma) $.
Also, $I(\gamma)$ does not contain any record, nor any site that belongs to the head or tail of another $ m $-soliton with $ m \geq k $.
Moreover, if $\gamma' \in \Gamma_m \eta $ with $m > k$ is such that $I(\gamma)\cap I(\gamma')\ne\varnothing$, then $I(\gamma)\subset I(\gamma')$.
If $\gamma$ and $\gamma'$ are two different $k$ solitons, then $I(\gamma)\cap I(\gamma')=\varnothing$.
\end{lemma}

\begin{figure}[b]
\centering
\rule{.95\textwidth}{.5pt}
\par
\bigskip
\(
\bf
\cdots
000000\underline{\orange11\underline{\blue01}111000\underline{\red1100}\underline{\red1100\underline{\blue10}}00}00\underline{\purple1111\underline{\green000111}0000\underline{\green111000}}0
\cdots
\)
\smallskip
\caption{\small%
Here we show $ I(\gamma) $ in an example with 9 records, a 5-soliton, a 4-soliton, two 3-solitons, two 2-solitons and two 1-solitons, with one color for each size.
In this example, a 1-soliton is contained in a 2-soliton, both 2-solitons are contained in the 5-soliton, both 3-solitons are contained in the 4-soliton.
$ I(\gamma) $ is underlined with the same color as $ \gamma $, and black zeros are records.
(color online)
}
\label{fig:intervals}
\end{figure}

The interval $ I(\gamma) $ and the above properties are illustrated in Fig.~\ref{fig:intervals}.

Recall the notation $\gamma^{1}$ from Proposition~\ref{prop:solitontrack}. The following lemma says that a soliton can move forward only if it is not already nested inside a larger one.
	
\begin{lemma}
\label{lem:xgamma}
Let $ \gamma \in \Gamma_k \eta $.
If $ I(\gamma) \subseteq I(\gamma') $ for some $ \gamma' \in \Gamma_m\eta $ with $ m > k $, then $\T(\gamma^{1})=\cH(\gamma)$ and $\cH(\gamma^{1})=\T(\gamma)$; hence $ x(\gamma^1) = x(\gamma) $.
Otherwise, $\T(\gamma^{1}) \ne \cH(\gamma)$ and $ x(\gamma^1) > x(\gamma) $.
\end{lemma}

Our last observation is that smaller solitons never overtake larger ones.
\begin{lemma}
\label{lemma:comparex}
Let $ \eta\in\cX $ and suppose that, $\gamma\in \Gamma_{k} \eta$ and $\tilde\gamma\in \Gamma_{m} \eta$ for some $ m \geq k \geq 1 $. If $x(\gamma)<x(\tilde\gamma)$, then $ x(\gamma^{t}) < x(\tilde\gamma^{t}) $ for all $ t\in\N $.
\end{lemma}

These three lemmas are also proved in \S\ref{sec:postponed}.

\subsection{Asymptotic speeds}

We use $\theta$ to denote shift operators on $ \Z $, its power set $ \mathcal{P}(\Z) $ and $ \{0,1\}^\Z $. Namely,
\begin{equation}
\label{def: theta}
\theta x = x-1
,
\quad
\theta A = \{\theta x : x\in A\}
,
\quad
\theta\eta(y):= \eta(\theta^{-1} y) \text{ for } \eta\in \{0,1\}^\Z
.
\end{equation}
Let $ \mathcal{B} $ denote the Borel $ \sigma $-field of $ \{0,1\}^\Z $.
We say that a probability measure $\mu$ on $\{0,1\}^\Z$ is \emph{shift-ergodic} if $\mu$ is $\theta$-invariant and for every event $A \in \mathcal{B}$ satisfying $\theta^{-1}(A)=A$ we have $\mu(A)=0$ or $1$.
Let $\mu$ be a shift-ergodic measure on $\cX$ and
denote by
$\rho_k$ the mean number of $k$-solitons per excursion,
by $w_0=1+\sum_k 2k\rho_k$ the mean distance between records, and
by $\orho_k = \frac{\rho_k}{w_0}$ the mean number of $k$-solitons per unit space (precise definitions in \S\ref{sub:finitew}).
Recall that $x(\gamma)$ is the leftmost site of a soliton $\gamma$ and that $\gamma^t$ is the soliton $\gamma$ at time $t$.
We now state the main result concerning soliton speeds.

\begin{theorem}
\label{thm:hspeed}
Let $\mu$ be a $T$-invariant and shift-ergodic measure on $\cX$.
Then there exists deterministic speeds $(v_k)_k$, 
such that, $\mu$-a.s., for all $k\ge 1$ and $\gamma \in \Gamma_k\eta$,
\begin{align}
\label{eq:speedexists}
\lim_{t\to\infty} \frac{x(\gamma^t)}{t} &= v_k.
\end{align}
The soliton speeds $ v_k $ are finite, positive, increasing in $ k $, and satisfy the system
\begin{equation}
\label{eq:hspeedseq}
v_k = k + \sum_{m<k} 2m \orho_m (v_k-v_m) - \sum_{m>k} 2k \orho_m (v_m-v_k)
.
\end{equation}
\end{theorem}

\begin{figure}[b]
\includegraphics[width=\textwidth]{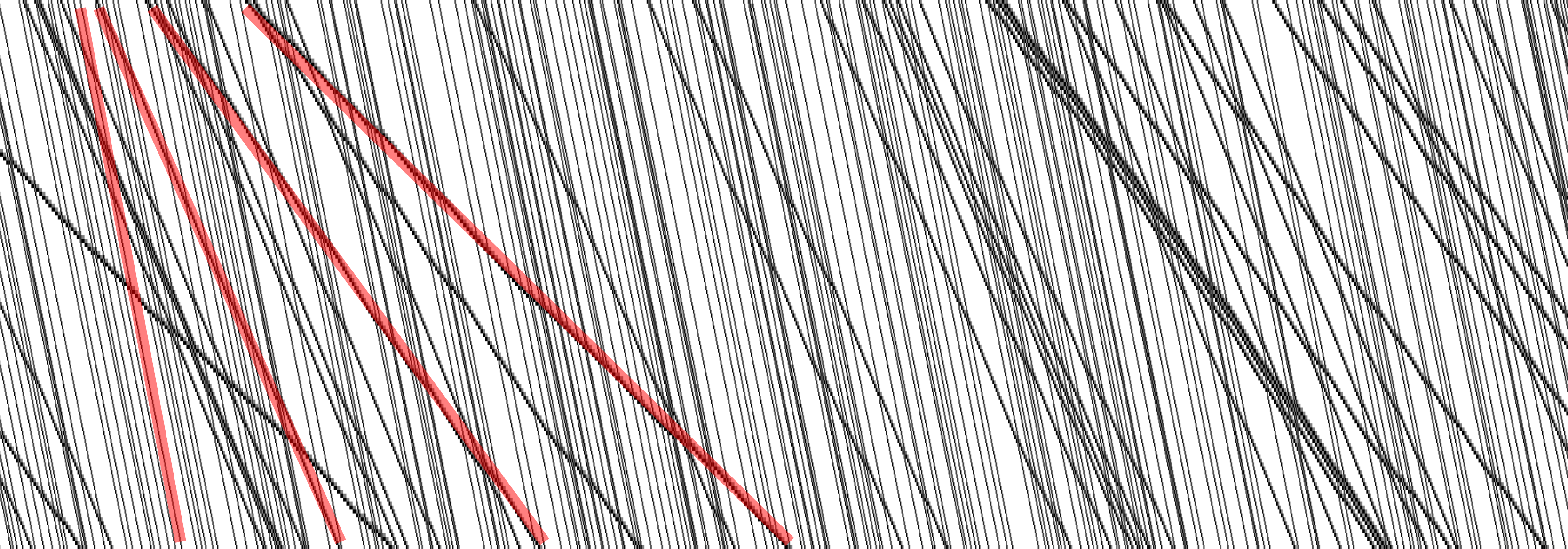}
\caption{\small%
Simulation for an i.i.d.\ configuration with density $0.15$.
The transparent red lines have deterministic slopes computed by Theorem~\ref{thm:speedsexplicit}, which have been manually shifted so that they would overlay a soliton.
This window covers 2000 sites and 140 time steps going downwards, and has been stretched vertically by a factor of $5$.
The figure in the first page is the same except for the density.
(high resolution, color online)}
\label{fig:speeds1}
\end{figure}

System~\eqref{eq:hspeedseq} comes from the following.
When a $k$-soliton is isolated, it advances by $k$ units, and when it encounters an $m$ soliton, the encounter causes it to advance $2m$ extra units if $m<k$ or 
stay put for $2$ units of time if $m>k$. 
The term $\orho_m |v_k-v_m|$ gives the frequency of such encounters as seen from a $k$-soliton.

When $ \rho_k=0 $, the soliton speed $ v_k $ does not come from~\eqref{eq:speedexists} but in principle formally from~\eqref{eq:hspeedseq}.
There is still an interpretation for $ v_k $ in terms of the dynamics, as discussed in \S\ref{sub:vertical}.

When $\rho_k>0$ for finitely many $ k $, the system has a unique solution~\cite[Lemma~5.1]{CroydonSasada20}.
We believe it is unique in general, but we have been unable to prove it.
The following partial result is proved in \S\ref{sec:postponed}.

\begin{proposition}
\label{prop:uniqueness}
If $\sum_m m \orho_m < \frac{1}{4}$, then the non-negative solution to~\eqref{eq:hspeedseq} is unique.
\end{proposition}

As a side remark, $ \sum_m m \orho_m = \lambda $, that is, the mean number of occupied boxes per unit space.
Uniqueness has not been proved to hold in general, and we show that $(v_k)_k$ is determined by the vector $ (\orho_k)_k $ under a stronger assumption in terms of soliton components (described in \S\ref{sub:invmeasures}).

\begin{theorem}
\label{thm:speedsexplicit}
If, when conditioned on having a record at $ x=0 $, $\mu$ has independent soliton components, and each component is i.i.d.,\
then the soliton speeds $(v_k)_k$ in~\eqref{eq:speedexists} are also given by the unique solution to
\begin{equation}
\label{eq:hspeedexplicit}
w_k=1+\sum_{m>k} 2(m-k) \rho_m
, \
\
\rho_k = \alpha_k w_k
, \
\
s_k = k+\sum_{m<k}2(k-m) s_m \alpha_m
, \
\
v_k = \frac {s_k}{w_k}
,
\end{equation}
and in particular they are determined by $(\rho_k)_k$.
\end{theorem}

In system~\eqref{eq:hspeedexplicit}, $w_k$ is the density of $k$-slots per excursion (see \S\ref{sub:components} for the definition of $k$-slot), $\alpha_k$ is the density of $k$-solitons per $k$-slot, $s_k$ is the average size of the head of a $k$-soliton, $k-m$ is the number of $m$-slots in the head of a $k$-soliton, and the factor $\frac{1}{w_k}$ is the probability that a typical $k$-soliton is free to move (see \S\ref{sec:speeds} for details).

The proof of~\eqref{eq:hspeedexplicit} uses independence properties of the components for an explicit computation of the mean jump size of a typical $k$-soliton $\gamma$ in one iteration.
By ergodicity, the mean jump size is $v_k$, the limit of $x(\gamma^t)/t$, as shown in \S\ref{sub:speedexists}.

In the setup of Theorem~\ref{thm:hspeed}, we cannot compute the mean jump size explicitly.
However, if we assume that the solution to~\eqref{eq:hspeedseq} is indeed unique, then by taking an initial measure with independent components and the same vector $ (\rho_k)_k $, we see that the vector $ (v_k)_k $ must be given by~\eqref{eq:hspeedexplicit}.
So if the solution to~\eqref{eq:hspeedseq} is indeed unique (as conjectured), the independence assumption in Theorem~\ref{thm:speedsexplicit} can be waived.

\begin{figure}[b]
\includegraphics[width=\textwidth]{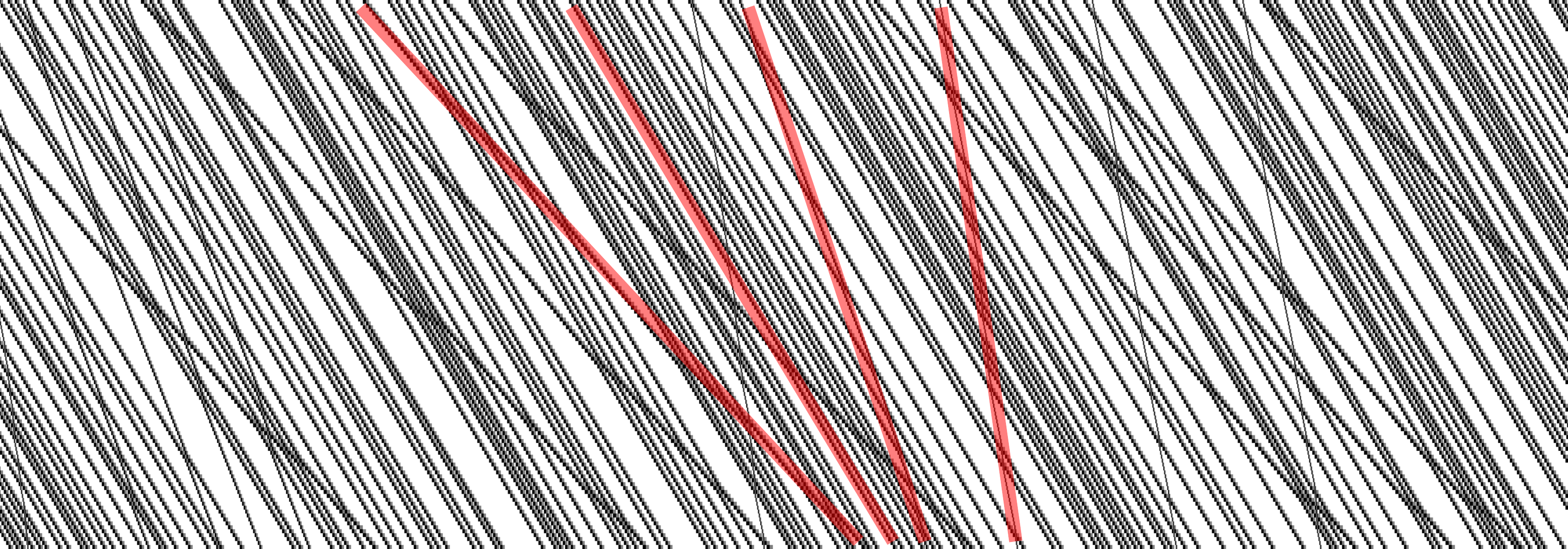}
\caption{\small%
Simulation for $(\rho_k)_k=(.006,.005 ,.1,.003,0,0,0,\dots)$.
The initial configuration was obtained by first appending one $k$-soliton with probability $\rho_k$ after each record, and then applying $T$ a number of times in order to mix.
As in Fig.~\ref{fig:speeds1} it is a 2000x140 window stretched by 5, and red lines are deterministic.
(high resolution, color online)}
\label{fig:speeds2}
\end{figure}

In practice, the soliton speeds $(v_k)_{k}$ can be computed by truncating $\rho$ (replace $\rho_k$ by $0$ for large $k$), and solving these finite recursions for $w$, $\alpha$, $s$ and finally $v$.
When the initial ball configuration consists of i.i.d.\ Bernoulli random variables, one can find $\alpha_k$ explicitly in terms of the density $\lambda$ by computing partition functions~\cite{FerrariGabrielli20}, substitute the equation for $\rho$ into that for $w$, and then compute $s$ and $v$.
Using this and the above theorem, we have found the asymptotic speeds of the solitons for the simulations shown in Figs.~\ref{fig:speeds1} and~\ref{fig:speeds2} as well as that of the first page.

\subsection{Slot decomposition}
\label{sub:components}

Recall that a $k$-soliton has a head and a tail, each one consisting of $k$ (possibly non-consecutive) sites.
We say the $j$-th box of the head or tail of an $m$-soliton is a \emph{$k$-slot} for all $k<j$ and that a record is a $k$-slot for every $k$.
Roughly speaking, the $k$-slots are the places where $k$-solitons can be appended, see \S\ref{sub:slot} for precise definitions and examples.

The set of configurations with a record at the origin are defined by
\begin{align}
\hcX:=\{\eta\in\cX:0\in R\eta\}
.
\end{align}
Assume that $\eta\in\hcX$.
Enumerate the $k$-slots from left to right in a way that the \emph{$0$-th $k$-slot} is at position $s_k(\eta,0)=0$, and let $s_k(\eta,i)$ denote the position of the $i$-th $k$-slot for $i\in\Z$.
We say that a $k$-soliton $\gamma$ is \emph{appended} to the $i$-th $k$-slot if its head and tail are contained
between $s_k(\eta,i)$ and $s_k(\eta,i+1)$.
Define the \emph{$k$-th component} of $\eta$ as the configuration $M_k\eta$
given by
\(
M_k\eta(i) := \text{number of the $k$-solitons appended to the $i$-th $k$-slot}.
\)
Define
\begin{align}
\cM:=\bigl\{\zeta= (\zeta_k)_{k\ge1}, \,\zeta_k \in (\Z_+)^\Z: \textstyle{\sum}_k \zeta_k(i)<\infty,\,\text{ for all }i\in\Z\bigr\}.
\end{align}
The next result, proved in \S\ref{sub:reconstruction}, says that we can recover $\eta$ from its components $(M_{k}\eta)_{k}$. 
\begin{theorem}
\label{bijection1}
For $\eta\in \hcX$, we have $(M_{k}\eta)_{k}\in \cM$. Moreover, the map $M:\eta\in \hcX\mapsto (M_k\eta)_k$ is invertible. 
\end{theorem}

The fundamental property of this decomposition is that it makes the BBS dynamics linear.
We show a simple case, deferring the full statement until Theorem~\ref{thm:hierarchy} in \S\ref{sec:compevol}.

\begin{theorem}
\label{simple-linear}
Suppose $ \eta \equiv 0 $ on $ \Z_- $ and $ \eta $ has infinitely many records on $ \Z_+ $.
Then, for all $k\ge 1$, 
\begin{align}
M_k T\eta =\theta^{-k} M_k\eta.
\end{align}
\end{theorem}

So, when we see a configuration $ \eta $ through its components, the $k$-th component is displaced by $k$ units, \emph{without interacting with the other components}.
In the general case, there will be larger solitons arriving from the left which will disrupt the enumeration of $k$-slots, making the statement substantially more involved.
See the details in \S\ref{sec:compevol}.

\subsection{Invariant measures}
\label{sub:invmeasures}

Using the slot decomposition and the reconstruction map $(M_{k}\eta)_{k}\mapsto \eta$ that will be described in detail in \S\ref{sub:reconstruction}, we get the following.

\begin{theorem}
\label{thm:invariant}
Let $\zeta=(\zeta_k)_{k\ge 1}$ be independent random elements of $(\Z_{+})^{\Z}$ with shift-invariant distributions satisfying $\sum_k k E[\zeta_k(0)] <\infty$ and $P(\sum_{k,i}\zeta_{k}(i)>0)=1$.
Then there exists a unique shift-invariant
probability $\mu$ on $\cX$ such that $M_k \eta \overset{d}{=}\zeta_k$ when $\eta$ is distributed with $\mu$ conditioned on $ \hcX $.
This measure $\mu$ is $T$-invariant.
If moreover $(\zeta_k(i))_{i\in\Z}$ is i.i.d.\ for each $k$, then $\mu$ is also shift-ergodic.
\end{theorem}
The above theorem says that the family of invariant measures for this dynamics is at least as large as the family of sequences of states of $k$-soliton configurations.
In particular, given a sequence $(\alpha_k)_k$ specifying the density of $k$-solitons in the $k$-th component,
we can construct an infinite number of mutually singular shift-invariant and $T$-invariant laws $\mu$ on $\cX$, all having the same specified component densities.

The extra assumption needed in Theorem~\ref{thm:speedsexplicit} is that $\mu$ be of the above form, i.e., conditioning on $ \hcX $, each $k$-th component is i.i.d.\ and they are independent over $k$.
In this case, we can also study the speed of tagged records and the speed of solitons measured in terms of tagged records, see \S\ref{sub:vertical}.
As pointed out above, this condition should not be necessary for Theorem~\ref{thm:speedsexplicit} to hold.

We should also note that the converse of Theorem~\ref{thm:invariant} is false. In particular, there exist invariant measures which are not constructed from independent components;
see Examples~\ref{exp1} and~\ref{exp2} as well as Remark~\ref{rmk:conjecture} in \S\ref{sub:muconstruct}.

\section{Slot decomposition}
\label{sec:solitons}

\begin{figure}[b]
\centering
\includegraphics[width=.8\textwidth]{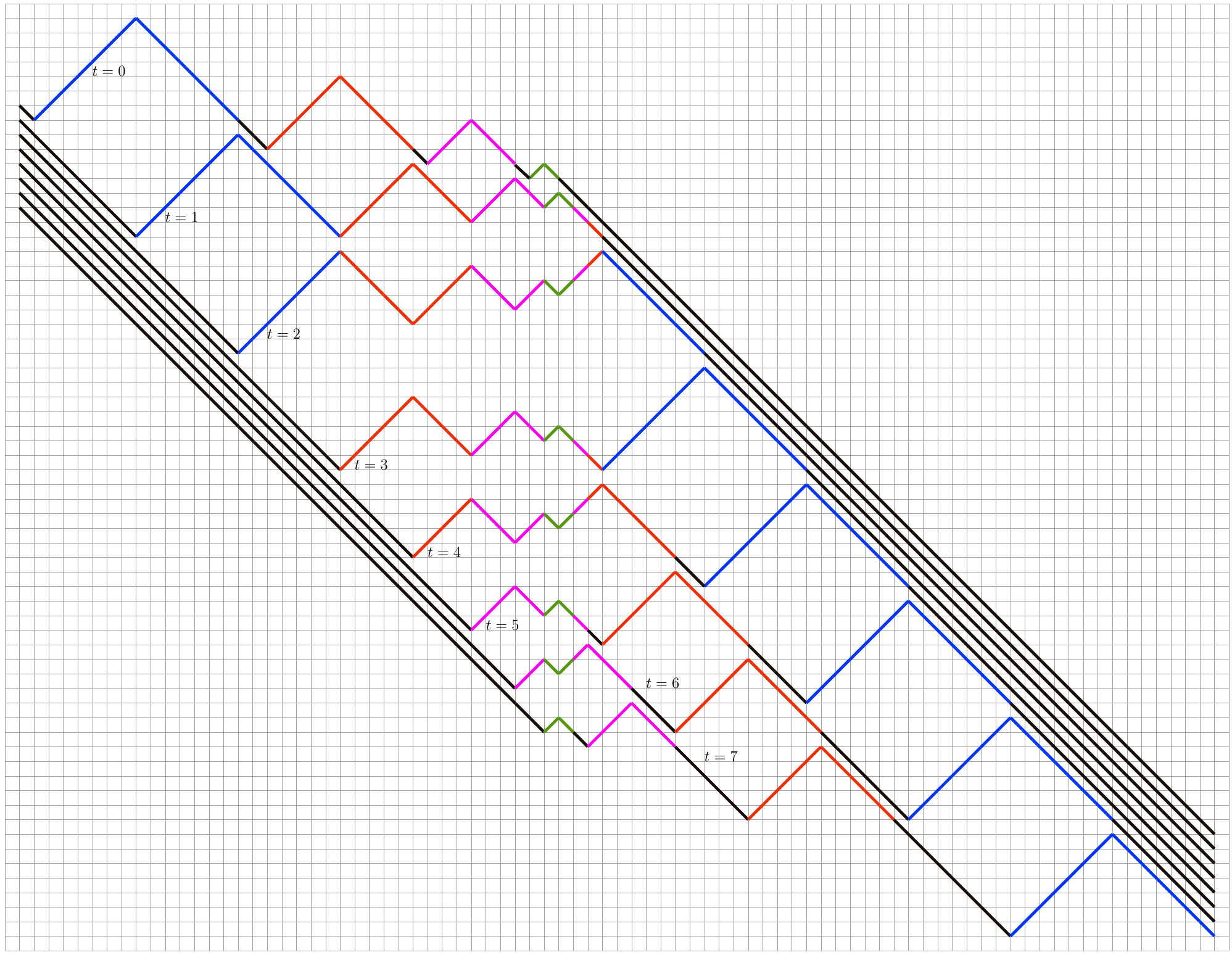}
\caption{\small%
Time-evolution of a walk under seven iterations of $T$.
This example has four solitons,
of size 7, 5, 3 and 1.
Different colors are used to highlight their conservation.
To facilitate view we have shifted the walk at time $t$ by $t$ units down. (color online)}
\label{fig:conservation}
\end{figure}

Let $\xi=\xi[\eta]$ be a walk on $\Z$ that jumps one unit up at $x$ when there is a ball at $x$ and jumps one unit down when box $x$ is empty.
That is,
\begin{align}
\xi(x)-\xi(x-1) = 2\eta(x)-1
.
\end{align}
Note that for each $\eta$, such a walk $\xi[\eta]$ is not unique and only defined up to a vertical shift. 
We define records for a walk $\xi$ in the usual sense, that is,
we say that $x$ is a \emph{record} for $\xi$ if
$\xi(z)> \xi(x)$ for all $z<x$.
Let $ \cW_\lambda $ the the set of all simple walks $ \xi $ such that $ \eta[\xi] \in \cX_\lambda $, which is the set of walks $ \xi $ such that
\begin{equation}
\label{eq:slope}
\lim_{x\to\pm\infty} \frac{\xi(x)}{x} = 2 \lambda - 1,
\end{equation}
and $\cW = \cup_{0< \lambda <1/2 }\cW_{\lambda}$.
Then every $\xi \in \cW$ satisfies
$\min\limits_{y\le x}\xi(y) \in \Z$ for all $x \in \Z$, and we define
\begin{align}
\label{eq:Txi}
T\xi(x) := 2 \min_{y\le x} \xi(y) - \xi(x) = \big[\min_{y\le x}\xi(y)\big]-\Big[\xi(x)-\min_{y\le x}\xi(y)\Big].
\end{align}
This amounts to reflecting the walk $\xi$ with respect to its running minimum.
It is worth mentioning that the above transformation for Brownian motion was studied by Pitman~\cite{Pitman75}.
The way that $T$ acts on $\cW$ is illustrated with an example in Fig.~\ref{fig:conservation}. 

One can see $\xi$ as a \emph{lift} of $\eta$ which includes an arbitrary choice of vertical shift, or equivalently an arbitrary labeling of records in increasing order.
Conversely, $\eta[\xi]$ is unambiguously defined by $\eta(x)=\frac{1+\xi(x)-\xi(x-1)}{2}$.
Consider the following diagram:
\begin{displaymath}
\xymatrixcolsep{3pc}
\xymatrix{
\xi
\ar[r]^{T}
\ar@/^.5pc/@{->}[d]^{\mathcal{P}}
& T\xi
\ar@/^.5pc/@{->}[d]^{\mathcal{P}}
\\
\eta
\ar[r]^{T}
\ar@/^.5pc/@{->}[u]^{\mathcal{L}}
& T\eta
\ar@/^.5pc/@{->}[u]^{\mathcal{L}}
}
\end{displaymath}
Notice that the above definition of record coincides with the one given in~\eqref{eq:reta} and~\eqref{t2} is equivalent to~\eqref{eq:Txi}. Therefore, this diagram commutes except that the lifting $\mathcal{L}$ misses uniqueness while the projection $\mathcal{P}$ cancels such non-uniqueness.
They are analogous to the derivative and indefinite integral where the latter comes with an indeterminate additive constant.
If a property is insensitive to the choice of the lift $\xi[\eta]$, then it is in fact a property of $\eta$, even if is described in terms of $\xi$.
For instance, 
for $ \lambda < \frac{1}{2}$, $ T $-invariance of $\cX_\lambda$ is equivalent to the $T$-invariance of $\cW_{\lambda}$, which follows immediately from~\eqref{eq:slope} and~\eqref{eq:Txi}.
Note that properties of $\eta$ always translate to $\xi$, for instance $\Gamma_m \xi$ means simply $\Gamma_m \eta[\xi]$, etc.
However, some of the objects considered in this section do depend on the lift $\xi$.

\subsection{Slots and components}
\label{sub:slot}

\begin{figure}[b]
\centering
\includegraphics[width=.32\textwidth,page=1]{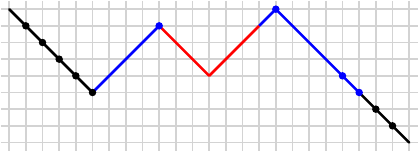}
\hfill
\includegraphics[width=.32\textwidth,page=2]{figures/slots-intro}
\hfill
\includegraphics[width=.32\textwidth,page=3]{figures/slots-intro}
\\
\includegraphics[width=.32\textwidth,page=4]{figures/slots-intro}
\hfill
\includegraphics[width=.32\textwidth,page=5]{figures/slots-intro}
\hfill
\includegraphics[width=.32\textwidth,page=6]{figures/slots-intro}
\caption{\small%
A red $ 3 $-soliton can be appended to a blue $ 5 $-soliton in $ 2 \times (5-3) = 4 $ different places, represented by blue dots.
It is also possible to append it to records, represented by black dots.
Attempting to insert a $ 3 $-soliton at a site not marked by a dot will 
result in erroneous soliton identification. For instance, the 3-soliton in the middle bottom plot 
should actually start 3 sites earlier, while in the right bottom plot we should have a 2 and 6-soliton instead of 3 and 5. 
The green crosses indicate that the coloring is inconsistent with the procedure shown in Fig.~\ref{fig:algo}.
(color online)}
\label{fig:slots-intro}
\end{figure}

\begin{figure}[b]
\centering
\includegraphics[width=.9\textwidth]{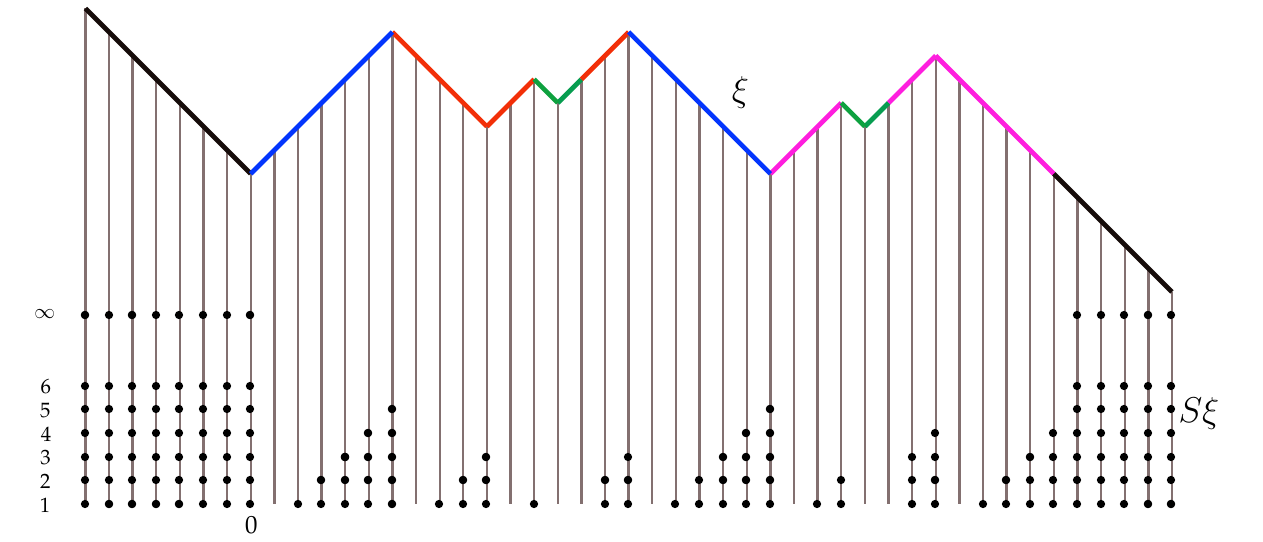}
\caption{\small%
Slot configuration of a walk $\xi$.
Different colors correspond to different solitons; records are painted in black. For each site, the number of dots below it indicates its level in slots: there being $k$ dots means a $k$-slot. (color online)}
\label{fig:slots1}
\end{figure}

\begin{figure}[b]
\centering
\includegraphics[width=.8\textwidth]{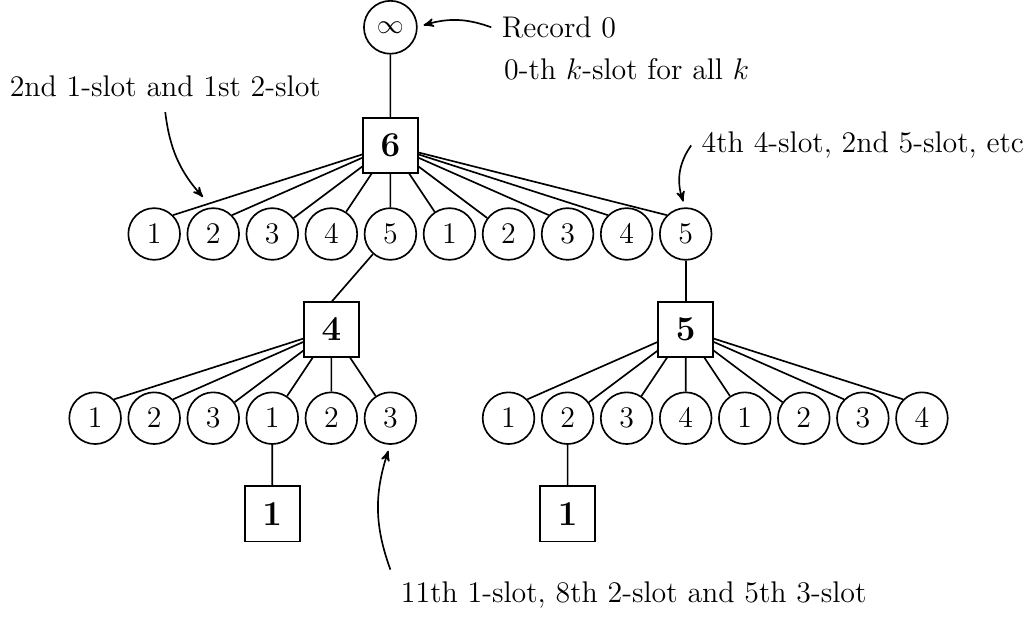}
\caption{\small%
An illustration of how the solitons are nested inside bigger solitons via slots, in the same sample configuration as in Fig.~\ref{fig:slots1}.
Solitons are represented by squares and slots by circles. For each $k \ge 1$, each slot with index $m\ge k$ is a $k$-slot.
We say it is the $n$-th $k$-slot, where $n$ is determined by counting how many $k$-slots appear before it in the depth-first order, and the counting starts from the 0-th $k$-slot present at Record~0.
}
\label{fig:slots-tree}
\end{figure}

We now describe how solitons can be nested inside each other via what we call \emph{slots}.
Intuitively, the idea of $k$-slot is that it marks a place where a $k$-soliton can be inserted without interfering with the rest of the configuration in terms of the Takahashi-Satsuma algorithm, see Fig.~\ref{fig:slots-intro}.

Let $\gamma\in \Gamma_{k}(\xi)$ be a $k$-soliton for the walk representation $\xi$.
We label the sites in the head (resp.~in the tail) of $\gamma$ from left to right: $\cH(\gamma)=\{\cH_1(\gamma), \dots, \cH_k(\gamma)\}$ (resp.~$\T(\gamma)= \{\T_1(\gamma),\dots,\T_k(\gamma)\}$).
The \emph{slot configuration} $S\xi:\Z\to \{0,1,2,\dots\}\cup\{\infty\}$ is defined by
\begin{align}
S\xi(x) :=
\begin{cases}
m-1,&\hbox{if } x = \T_m(\gamma) \text{ or } \cH_m(\gamma), \hbox{ for }\gamma \in\Gamma_k\xi \hbox{ and } k \ge m,\\
\infty,&\hbox{if } x \hbox{ is a record for } \xi.
\end{cases}
\end{align}
For each $k\ge 1$ we say that $x$ is a \emph{$k$-slot} for $\xi$ if $S\xi(x)\ge k$.

Note that a record is a $k$-slot for all $k$, and an $m$-soliton contains a number $2m-2k$ of $k$-slots, see Fig.~\ref{fig:slots1}.
Since every $\xi\in \cW$ has infinitely many records, it also has infinitely many $k$-slots.

For $j\in \Z$, the position of the record at level $-j$ will be called \emph{Record}~$j$ and denoted as
\begin{align}
\label{eq:r111}
r(\xi,j) := \min\{x\in\Z: \xi(x)=-j\}.
\end{align}
This is the leftmost site where the walk $\xi$ takes the value $-j$.
If $\xi\in\cW$, we have $r(\xi,j) \in \Z$ well-defined for all $j\in\Z$.
The site $r(\xi, 0)$ will play a central role in the sequel.
Note that $T\xi(x) \le \xi(x)$ by~\eqref{eq:Txi}, so $r(T\xi, j)\le r(\xi, j)$ for all $j$.

We label all the $k$-slots in increasing order: $\cdots<s_{k}(\xi, -1)<s_{k}(\xi, 0)<s_{k}(\xi, 1)<\cdots$, where $s_k(\xi,0):=r(\xi,0)$.
The set of $k$-slots of $\xi$ is denoted $S_k\xi:= \{s_k(\xi,i):{i\in\Z}\}$.
We then say that a $k$-soliton $\gamma$ is \emph{appended} to the $i$-th $k$-slot $s_{k}(\xi, i)$ if $\gamma\cap [s_k(\xi,i),s_k(\xi,i+1)-1]\ne \varnothing$.
Observe that if that is the case, we necessarily have $\gamma\subset [s_k(\xi,i)+1,s_k(\xi,i+1)-1]$, as a consequence of Lemma~\ref{lem:basics} and the fact that $k$-slots can only be records or sites of solitons with larger size.
It follows that each $k$-soliton in $\xi$ is appended to a unique $k$-slot. On the other hand, it is possible to have multiple $k$-solitons appended to a single $k$-slot.
Finally, we let $M_k\xi(i)$ be the number of $k$-solitons appended to the $i$-th $k$-slot and call $M_{k}\xi=(M_{k}\xi(i))_{i\in \Z}$ the \emph{$k$-th component} of $\xi$. For instance, in the example of Fig.~\ref{fig:slots1}, we have
\(
M_6\xi(0) = 1,
\
M_5\xi(2) = 1,
\
M_4\xi(2) = 1,
\
M_1\xi(9) = 1,
\
M_1\xi(18) = 1,
\)
and $M_k\xi(i) = 0$ otherwise.
See also Fig.~\ref{fig:slots-tree} for an illustration on how the solitons are nested inside each other via slots.

\subsection{Reconstructing the configuration from the components}
\label{sub:reconstruction}

\begin{figure}[b]
\centering
\includegraphics[clip,width=0.9\textwidth]{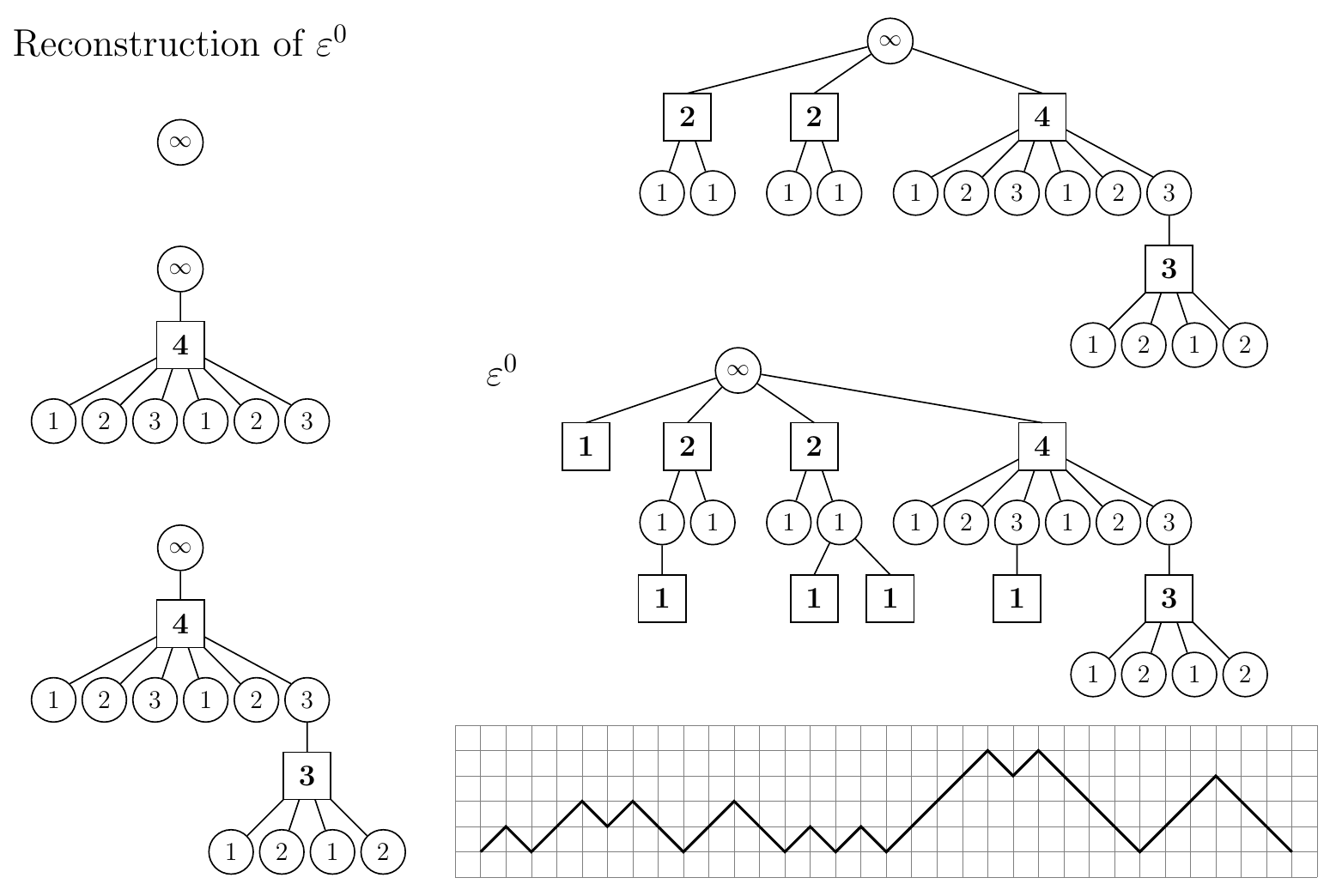}
\par
\vspace{1em}
\includegraphics[clip,width=0.9\textwidth]{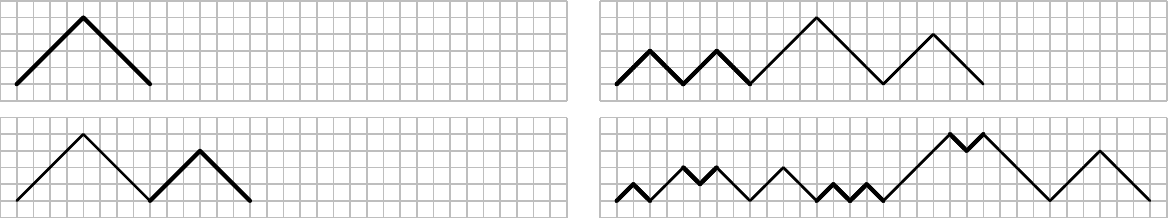}
\caption{\small%
Reconstruction algorithm for a single excursion. This example is obtained using the field $\zeta$ shown in Fig.~\ref{fig:zeta-2}.}
\label{fig:reconstruction-3}
\end{figure}

\begin{figure}[b]
\centering
\includegraphics[clip,width=0.9\textwidth]{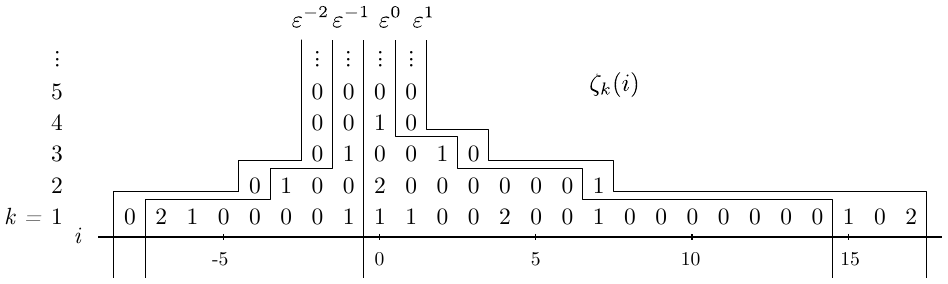}
\par
\includegraphics[clip,width=0.95\textwidth]{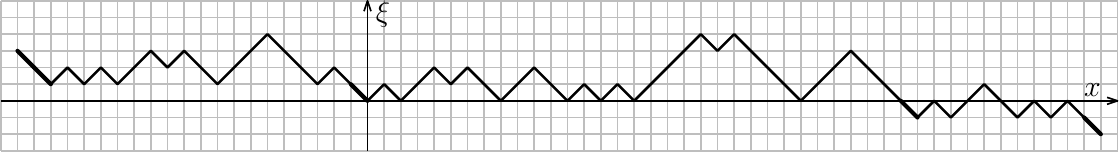}
\caption{\small%
Reconstruction of $\xi$ from $\zeta$.
In the lower part we show Records~$-2$ to~$2$ in boldface and the excursions between them.
Above we show the parts of the field $\zeta$ that used in the reconstruction of $\vep^{-2},\vep^{-1},\vep^{0},\vep^{1}$.
Reconstruction of $\vep^0$ was shown in Fig.~\ref{fig:reconstruction-3} and $\vep^{1},\vep^{-1},\vep^{-2}$ is shown in Fig.~\ref{fig:reconstruction-4}.}
\label{fig:zeta-2}
\end{figure}

\begin{figure}[b]
\centering
\includegraphics[clip,width=0.9\textwidth]{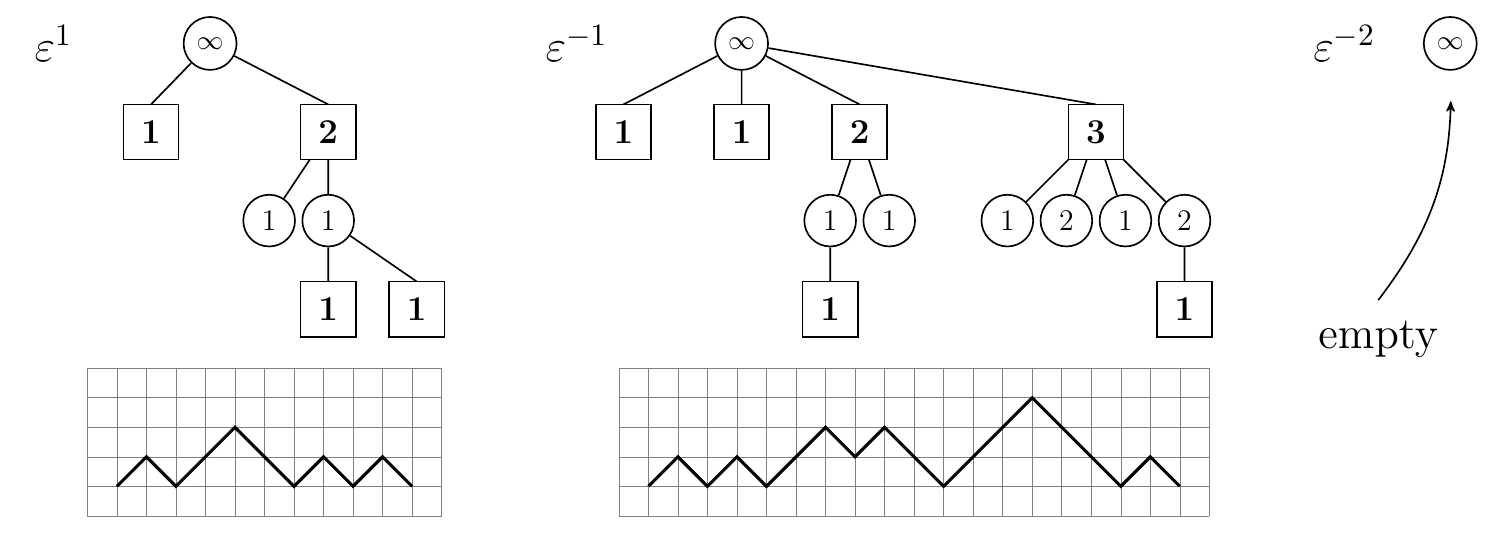}
\caption{\small%
Reconstruction algorithm for other excursions. The procedure is the same as in Fig.~\ref{fig:reconstruction-3} but all the intermediate steps are omitted.}
\label{fig:reconstruction-4}
\end{figure}

Here we prove Theorem~\ref{bijection1}. 
In this subsection only, we will work with a larger set of configurations than $\mathcal W$.
Let $\mathcal W_{*}$ be the set of walks $\xi$ such that $r(\xi, j)$ is well-defined for all $j\in \Z$.
Even though $\mathcal W_{*}$ is not preserved by $T$, the Takahashi--Satsuma algorithm still applies, as all the excursions of $\xi\in\mathcal W_{*}$ are finite.
So the $k$-th component $M_{k}$ described in \S\ref{sub:slot} is also defined for all $\xi\in \mathcal W_{*}$. On the other hand, since the slot decomposition $\xi \mapsto M\xi$ is insensitive to horizontal shifts, it is not possible to determine $\xi$ knowing $(M_k\xi)_{k\ge 1}$. So we introduce
\begin{align}
\hcW_{*}:= \{\xi\in\cW_{*}: r(\xi,0)=0\}.
\end{align}
({As a side remark, $ \cW_* $ is the set of simple walks $ \xi $ such that $ \lim_{x\to-\infty} \xi(x)=+\infty $ and $ \liminf_{x\to+\infty} \xi(x) = -\infty $.
We can build a configuration $\xi\in\cW_*$ with $ T\xi\not\in\cW_* $ by appending a soliton of height $2^n$ between Records $ -n $ and $ -n+1 $, so that $ \liminf_{x\to-\infty} T\xi(x) = -\infty $.})

Note that, unlike the lift $\xi[\eta]$ from $\cX$ to $\cW$ which was not unique, for $\eta$ in $ \hcX $
there is a unique lift $\xi[\eta]$ which is in $\hcW_{*}$.
We denote this unique lift by $\bxi[\eta]$. So the maps $ \eta \mapsto M_k \eta $ as seen in Theorem~\ref{bijection1}
are given by $M_k\eta = M_k\bxi[\eta]$.
But to remain consistent with the previous subsections, we continue to work with $\xi$ instead of $\eta$.

Denote by $M$ the map $\xi \mapsto M\xi :=(M_k\xi)_{k\ge 1}$.
We now show that $M: \hcW_{*}\to \mathcal M$ is invertible.

The \emph{height} of the excursion between Record~$j$ and Record~$j+1$, denoted as $m(j)$, is defined as $0$ if the excursion is empty or as the largest $k$ such that a $k$-soliton is contained in the excursion.
Denote also $i_k(j)$ the label of the $k$-slot located at Record~$j$.
Since $k$-slots can only be created by solitons of larger sizes, we note that
\begin{align}
\label{eq:height}
m(j) &=
\min\{k\ge 0: M_{k'}\xi(i_{k'}(j))=0 \hbox{ for all }k'>k\} \in \{0,1,2,\dots\}.
\end{align}
Both $m(j)$ and $i_k(j)$ depend on $\xi$ but we omit it in the notation.
Denote the slot decomposition $M\xi$ of $\xi$ by $\zeta$, that is,
\[
\zeta = (\zeta_k)_k = M \xi
.
\]
Note that, for each $\xi \in \cW_{*}$ and $j\in\Z$, $i_k(j)=r(\xi,j)$ for all $k\ge \max\{m(0), \dots, m(j)\}$.
Since $m(j)$ is finite, we have that
\begin{equation}
\label{eq:finitezeta}
\text{ for each } i\in\Z,\
\zeta_k(i) = 0
\text{ for all large }
k.
\end{equation}
Namely, $\zeta\in \mathcal M$.
Conversely, suppose $\zeta = (\zeta_k)_{k\ge 1} \in \mathcal M$, so in particular $\zeta$ satisfies~\eqref{eq:finitezeta}.
We first give an algorithm which permits to reconstruct the excursion of $\xi$ between Records 0 and 1. 
To that end, we introduce the following notation: denote the number of $k$-slots in the excursion $\vep$ between successive records $y_1<y_2$ by
\begin{align}
\label{sup-vep}
n_k(\vep) := 1 + |S_k \xi \cap \{y_1+1,\dots, y_2-1\}|,
\end{align}
where the term $1$ refers to the record $y_1$ preceding $\vep$ and the second term counts the number of $k$-slots belonging to $m$-solitons of $\vep$ with $m>k$.
Here is the algorithm:

\begin{algorithm}[H]
\label{alg:reconstruct}
Let $\vep$ be the empty configuration. \\
Let $m:= \min\{k\ge0: \zeta_{k'}(0)=0 \hbox{ for all }k'>k\}$ \\
\For{$k = m,m-1,\dots,2,1$}{
Let $n_k := n_k(\vep)= \#\{x \in S_k \vep : r(\vep,0)\le x<r(\vep,1) \}$, as in~\eqref{sup-vep} \\
\For{$i=0,1,\dots,n_k-1$}{
Insert a number $\zeta_{k}(i)$ of $k$-solitons in the $i$-th $k$-slot of $\vep$, that is, to the right of site $x=s_k(\vep,i)$; boxes to the right of $x$ are shifted further right in order to accommodate the insertion of these $k$-solitons \\
This produces an updated configuration $\vep$ \\
}
}
\end{algorithm}

Note that $m$ is well-defined by~\eqref{eq:finitezeta}.
In case $m=0$, the algorithm produces an empty configuration. 
Let us also note that when inserting solitons in the above algorithm, there is only one way to do it consistently with the soliton decomposition: if a soliton is inserted to the right of the site $x$ with $\eta(x)=0$, then it has its head on the left and tail on the right (i.e.~$11\cdots100\cdots0$); otherwise, it is inserted with its tail on the left and head on the right (i.e.~$00\cdots011\cdots1$). 
The procedure is illustrated in Fig.~\ref{fig:reconstruction-3}.
Call $\vep^0$ the excursion between Record~$0$ and Record~$1$ obtained from $ \zeta $ as just described.
Construct $\vep^1$, the excursion between Records~1 and~2, using the same algorithm but with the data $\zeta^1 = (\zeta^1_k)_{k\ge 1}$, where each component is given by
\[
\zeta^1_k = \big( \zeta_k(n_k+i) \big)_{i \geq 0},
\]
which consists of the entries of $\zeta$ with non-negative indices $i$ not used in the reconstruction of $\vep^0$.
Note that $\zeta^1$ also satisfies~\eqref{eq:finitezeta}.
Iterate this procedure to construct an infinite sequence of excursions $(\vep^j)_{j=0,1,2,\dots}$.
See Fig.~\ref{fig:reconstruction-4}.

To reconstruct the configuration to the left of Record~$0$, that is, to obtain the excursions $\vep^j$ with negative $j$, we use an analogous algorithm that uses the entries of $\zeta$ with $i$-indices starting at $-1$ and moving left instead of starting at $0$ and moving right.
First take $\zeta^{-1} = (\zeta^{-1}_k)_{k \geq 1}$ where each component is given by $\zeta^{-1}_k = \big( \zeta_k(i) \big)_{i < 0}$ and use $\zeta^{-1}$ to construct $\vep^{-1}$.
Then define $\zeta^{-2} = (\zeta^{-2}_k)_{k \geq 1}$ where each component is given by $\zeta^{-2}_k = \big( \zeta_k(i-n_k) \big)_{i < 0}$ and use it to construct $\vep^{-2}$.
Iterate this procedure to construct an infinite sequence of excursions $(\vep^j)_{j=-1,-2,\dots}$.

Put Record~$0$ at the origin and concatenate the excursions with one record between each pair of consecutive excursions.
This yields a walk denoted as $\xi^*$, shown in Fig.~\ref{fig:zeta-2}. Note that in $\xi^{*}$, all the excursions are finite and therefore $r(\xi, j)$ is finite for all $j\in \Z$. Namely, $\xi^{*}\in\hcW_{*}$.

Call $M^{-1}:\zeta \mapsto \xi^*$ the resulting transformation.
We claim that $M^{-1}$ is the inverse map of $M$, that is, $M^{-1} M\xi = \xi$ for $\xi\in\hcW_{*}$ and $MM^{-1}\zeta=\zeta$ for $\zeta\in \mathcal M$.
The second identity follows from the fact that $M_{k}\xi^{*}=\zeta_{k}$ for each $k$.
Now let $\vep$ be an excursion of $ \xi $.
For $k\ge 0$, denote by $\vep_{[k]}$ the ball configuration obtained by removing all the boxes belonging to an $\ell$-soliton with $\ell\le k$.
Then $\vep_{[0]}=\vep$ and $\vep_{[k]}$ is the empty excursion for $k$ sufficiently large. Now observe from the previous definitions that $M_m\vep_{[k]}=M_m \vep$ for all $m>k$, because $m$-slots are only created by solitons of sizes larger than $m$.
So the reconstruction algorithm correctly finds $\vep_{[k]}$ from $M_{k}\vep$ and $\vep_{[k+1]}$, hence it correctly finds $\vep$.

This shows that $M: \hcW_{*}\to \mathcal M$ is invertible, and so is its restriction to a subset $M: \hcW\to M(\hcW)\subset \mathcal M$, where $\hcW:=\{\xi[\eta]\in \hcW_{*}: \eta\in \hcX\}$.
This proves Theorem~\ref{bijection1}.

Let us also point out that the image set $M(\hcW)$ is not simple to characterize. However, as we will see below, if we sample $\zeta$ with an appropriate measure, the resulting (random) element $M^{-1}\zeta$ does belong a.s.\ to $\hcW$.

\section{Evolution of components}
\label{sec:compevol}

\begin{figure}[b]
\centering
\includegraphics[width=.98\textwidth]{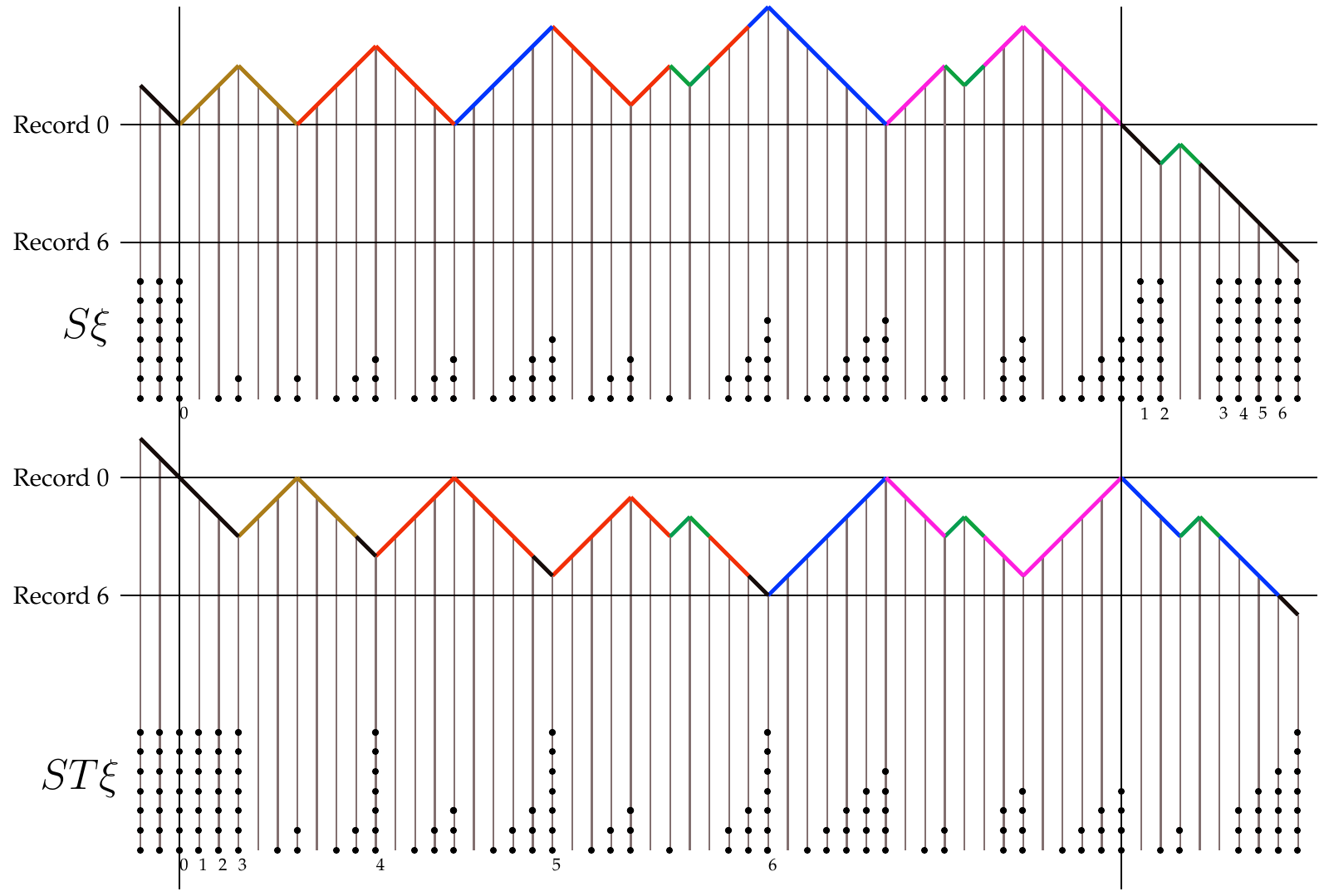}
\caption{\small%
We depict $T\xi$ below $\xi$. This example illustrates the various situations discussed in the proof of Theorem~\ref{thm:hierarchy}. For instance, the $5$-soliton (in pink) stays put and is nested in the $6$-soliton (in blue).
The displacement of the 6-soliton brings 5 ``new'' 5-slots to the left of the 5-soliton.
(color online)}
\label{fig:movingslots}
\end{figure}

Here we prove a stronger version of Theorem~\ref{simple-linear}. 
Recall that we can track a tagged soliton $\gamma$ after $t$ iterations of $T$ by~\eqref{eq:trackgamma}.
In order to describe how the BBS dynamics act on the components,
we introduce the flows of solitons as follows.
First, let
\begin{align}
\label{eq:solitonflow}
J^t_m \xi := \# \big\{ \gamma\in\Gamma_m\xi : \gamma\subset(-\infty,r(\xi,0))\hbox{ and } \gamma^t \subset[r(T^t\xi,0),\infty)\big\}.
\end{align}
In words, $J^{t}_{m}\xi$ counts the number of $m$-solitons which were to the left of Record $0$ initially and are found to be on the right of Record $0$ at time $t$. 
We say such solitons \emph{cross Record $0$.} 
Note that this is not the same as ``crossing the origin'' because Record~$ 0 $ is itself moving left.

Note that $ J^1_m \xi=0 $ for all large $ m $.
Indeed, using Proposition~\ref{prop:solitontrack} one can see that no more than $ r(\xi,0)-r(T\xi,0) $ solitons can cross Record $0$ at time $1$, 
so $ \sum_m J^1_m \xi < \infty $.
By the same argument, $ \sum_m J^2_m \xi < \infty $, and so on.

We now define an observable $o^t_k(\xi)$ which counts the flow of $k$-slots through Record~$0$
after $t$ iterations of $T$.
Because each $m$-soliton crossing Record~$0$ from left to right carries $2(m-k)$ $k$-slots, we define
\begin{align}
\label{itk}
o^t_k(\xi) := \sum_{m>k} 2(m-k) J^t_m \xi<\infty.
\end{align}

Using the observable $o^t_k(\xi)$, we define the \emph{tagged $0$-th $k$-slot at time $t$} by
\begin{align}
\label{eq:stk}
s^t_k(\xi,0) &:= s_k(T^t\xi,o^t_k(\xi)),
\end{align}
which is the position of the $o_k^t$-th $k$-slot counting from Record~$0$ of $T^t\xi$.
More generally, the \emph{tagged $i$-th $k$-slot at time $t$} is defined as
\begin{align}
s^t_k(\xi,i) &:= s_k(T^t\xi,o^t_k(\xi)+i).
\end{align}

We now state one of the central results of this paper.
In some sense, it shows that the action of $ T $ on the component configuration $ (M_k \xi)_{k \in \N} $ has a simple form, analogue to a Jordan form or upper triangular form for a linear map.
The $k$-th component is shifted by $k$ units, plus possibly a number resulting form the re-indexing of $k$-slots caused by larger solitons crossing Record~$ 0 $.
\begin{theorem}
\label{thm:hierarchy}
Let $\xi\in\cW$ and $t\in\{1,2, 3, \dots\}$. Then, for any $k$-soliton $\gamma$ of $\xi$,
\begin{align}
\label{eq: slot-shift}
\# \Big\{ {i\in\Z} : \gamma \subset \big( -\infty,s_k(\xi,i) \big) \hbox{ and } \gamma^t \subset \big[s^t_k(\xi,i),\infty\big) \Big\} = kt,
\end{align}
that is, between times $0$ and $t$ a tagged $k$-soliton moves $kt$ unit to the right in terms of tagged $k$-slots, 
or equivalently, the right-to-left flow of tagged $k$-slots through a tagged $k$-soliton is exactly $kt$. 
For each $k\in \N$, the $k$-soliton component of $T^t\xi$ is a shift of the $k$-soliton component of~$\xi$:
\begin{align}
\label{eq:M_kdynamic}
M_k T^t \xi(i) = M_k\xi\big(i - o^t_k(\xi)-kt\big)
,
\end{align}
where $o^{t}_{k}(\xi)$ is defined in~\eqref{itk}.
Moreover, the offset $o_k^t(\xi)$ is determined by $(M_m \xi)_{m>k}$, so that $M_{k}T^{t}\xi$ is a function of $(M_m \xi)_{m \geq k}$.
\end{theorem}

Another way to interpret the theorem is based on the notion of \emph{$k$-bearer}, which we introduce now. 
Let $\pi \in S_k\xi \subseteq \Z$.
Then $\pi = s_k(\xi,j)$ for some $j$, and we define the corresponding $k$-bearer at time $t$ to be 
\begin{equation}
\label{eq:trackslot}
\pi^{k,t} := s_k \big( T^t\xi,o^t_k(\xi)+kt+j \big)
.
\end{equation}
An immediate consequence of Theorem~\ref{thm:hierarchy} is as follows: a $k$-soliton $\gamma$ is appended to the $k$-slot located at $\pi=s_k(\xi,j)$ if and only if $\gamma^{t}$ is appended to the $k$-slot located at $\pi^{k, t}$.
In particular, since the definition of a tagged soliton depends only on $\eta[\xi]$ according to Proposition~\ref{prop:solitontrack}, this implies that $\pi^{k,t}$ also only depends on $\eta[\xi]$.
We note that the difference between the tagged slots introduced before and bearers introduced just now is the factor of $kt$ related to the motion of $k$-solitons.
So a tagged $k$-soliton crosses $k$ tagged $k$-slots per unit time whereas it does not cross $k$-bearers (in fact it just follows one of them).
As a side remark, the notion of $k$-bearer allows us to define the soliton speed $v_{k}$ even if $\rho_{k}=0$ by replacing $x(\gamma^{t})$ with $\pi^{k, t}$ in~\eqref{eq:speedexists}.

\begin{proof}[Proof of Theorem~\ref{thm:hierarchy}]
We start by showing how the second statement follows from the first one.
Since every $k$-soliton crosses exactly $k$ tagged $k$-slots at each step,
the number of $k$-slots between any pair of tagged $k$-solitons is conserved by $T$.
Hence the $k$-th component as seen from the tagged $0$-th $k$-slot just shifts $k$ $k$-slots per unit time, while the term $o^t_k(\xi)$ accounts for the relabeling of $k$-slots caused by bigger solitons crossing Record~0.

For the proof of the first statement, it suffices to consider the case $t=1$. Let us first assume that there is only a finite number of balls in $\xi[\eta]$. Without loss of generality, we can also assume $r(\xi, 0)=0$. For $x\in \Z_{+}$, let us define $D_{k}(\xi, x)=\# S_{k}T\xi\cap [0, x] - \# S_{k}\xi\cap [0, x]$, namely the change in the number of $k$-slots found in $[0, x]$ after one iteration of $T$.
Let us recall the notation $\cH_{i}(\gamma)$ which stands for the $i$-th site of the head of a soliton $\gamma$.
We introduce the set $\mathcal D_{k}(\xi, x)=\{z\in [0, x]: \exists\,\gamma \text{ s.t. } z=\cH_{i}(\gamma), 1\le i\le k, \text{ and } \T_{i}(\gamma^{1})>x\}$.
Let us observe that if such a soliton $\gamma$ exists, then necessarily $x(\gamma^{1})> x(\gamma)$; otherwise, $\gamma$ merely swaps its head and tail, and we will have $\T_{i}(\gamma^{1})=\cH_{i}(\gamma)$ in that case.
Therefore, $\#\mathcal D_{k}(\xi, x)$ counts the number of slots in $[0, x]$ which have order $<k$ and belong to the heads of solitons that move to the right of $x$ after one iteration.
See Fig.~\ref{fig:movingslots} for an example.
We claim that for all $x\in \Z_{+}$ and all $k\ge 1$,
\begin{equation}
\label{pf-thm6}
D_{k}(\xi, x)=\# \mathcal D_{k}(\xi, x).
\end{equation}
To prove this, let us first suppose that $\eta$ consists of one soliton, say at the sites $[0, 2m)$.
Then the slot configuration $ST\xi$ is obtained from $S\xi$ by swapping the labels of the sites in $[0, m-1]$ with those in $[2m, 3m-1]$.
One readily checks that in this case $D_{k}(\xi, x)$ is tent shaped: linear with slope $1$ on $[0, k\wedge m-1]$, constant on $[k\wedge m-1, 2m]$ and null on $[2m+m\wedge k, \infty)$. This can be similarly checked for $\mathcal D_{k}(\xi, x)$ and the claimed identity holds in this case.

For a general configuration $\xi$, the sites in $\mathcal D_{k}(\xi, x)$ do not necessarily become records in $T\xi$, as they can be claimed by other solitons. However, we will see that even if this happens, it will only affect where the ``new'' $k$-slots are located but not the quantity $D_{k}(\xi, x)$. Proceeding with the proof of~\eqref{pf-thm6}, let us suppose that it holds for configurations with up to $p\ge 1$ solitons and let $\xi[\eta]$ be a configuration with $p+1$ solitons. Let $m\ge 1$ be the smallest size of the solitons in $\xi$ and let $\gamma$ be its leftmost $m$-soliton. We note that $\cH(\gamma)$ and $\T(\gamma)$ is back-to-back in $\xi$ (i.e~no other solitons lodged inside $\gamma$) as a consequence of Lemma~\ref{lem:basics}. Then we can write $I(\gamma)=\cH(\gamma)\cup\T(\gamma)=[a, a+2m)$ with $a=x(\gamma)$. Let us introduce $\tilde\xi(x)=\xi(x+2m)$ for all $x\ge a$ and $\tilde\xi(x)=\xi(x)$ otherwise. Namely, $\tilde\xi$ is the configuration obtained by removing $\gamma$. By noting that the operator $T$ consists in flipping portions of $\xi$ between consecutive records, we deduce that
\[
T\xi(x)=T\tilde\xi(a-1)+\sum_{a\le y\le x}(1-2\eta(y)), \quad \text{ if } x\in [a, a+2m),
\]
and $T\xi(x)=T\tilde\xi(x-2m)$ if $x\ge a+2m$, and $T\xi(x)=T\tilde\xi(x)$ otherwise.
In words, $T\xi$ is obtained by inserting a ``flipped'' version of $\gamma$ into $T\tilde\xi$ at site $a$.

To check~\eqref{pf-thm6}, let us first note that if $x(\gamma^{1})=x(\gamma)$, then $ST\xi$ is simply obtained by inserting the slot configurations of the sites in $\gamma^{1}$ into $ST\tilde\xi$, so that $D_{k}(\xi, x)=D_{k}(\tilde\xi, x)$ for $x<a$, $D_{k}(\xi, x+2m)=D_{k}(\tilde\xi, x)$ for $x\ge a$ and $D_{k}(\xi, x)=D_{k}(\tilde\xi, a-1)$ for $a\le x< a+2m$.
In other words, putting $\gamma$ back does not modify $D_{k}(\xi, x)$ in real terms.
Similar relationship holds between $\mathcal D_{k}(\xi, x)$ and $\mathcal D_{k}(\tilde\xi, x)$. In that case,~\eqref{pf-thm6} follows from the induction hypothesis.

Now suppose $x(\gamma^{1})\ne x(\gamma)$. From Lemma~\ref{lem:xgamma} and the fact that $\gamma^{1}$ is the smallest soliton in $T\xi$, we deduce the only possibility to be $\gamma^{1}=[a+m, a+3m)$ with $\cH(\gamma^{1})$ preceding $\T(\gamma^{1})$. Appealing again to the conservation of solitons, we see that the configuration $T\xi$ is obtained by inserting an $m$-soliton into $T\tilde\xi$ at the site $a+m$. It follows that $ST\xi$ is obtained by inserting the slot configurations of the sites in $\gamma^{1}$ into $ST\tilde\xi$ at the site $a+m$.
In particular, the slot configurations of $T\xi$ and $T\tilde\xi$ coincide up to the site $a+m$. For $x\in [a+m, a+2m)$, note that
$D_{k}(\xi, x)=D_{k}(\xi, a+m-1)+\# S_{k}T\xi\cap [a+m, x)- \# S_{k}\xi\cap [a+m, x)=D_{k}(\xi, a+m-1)$, as the slot configuration in $[a+m, a+2m)$ is unchanged. On the other hand, one can readily check $\mathcal D_{k}(\xi, x)=\mathcal D_{k}(\xi, a+m-1)$. Finally, let us check~\eqref{pf-thm6} for $x\in [a+2m, a+3m)$. Writing $i=x-a-2m+1$ and letting $A_{x}$ be the event $S\xi(x)< k$, we note that
$D_{k}(\xi, x)-D_{k}(\xi, x-1)=-\one\{i\le k\}+\one\{A_{x}\}$.
On the other hand, when comparing $\mathcal D(\xi, x)$ and $\mathcal D(\xi, x-1)$, we see that if $i\le k$, then $\mathcal D(\xi, x)$ loses one element, namely $z=\cH_{i}(\gamma)$, since $x=\T_{i}(\gamma^{1})$; but if $A_{x}$ occurs, then $x=\cH_{j}(\tilde\gamma)$ for some other soliton $\tilde\gamma$ and $j\le k$, so that $x\in \mathcal D_{k}(\xi, x)$.
Putting all the pieces together, we conclude that~\eqref{pf-thm6} also holds in this case, which completes its proof.

To see why~\eqref{pf-thm6} leads to the first statement in the theorem, let $\gamma$ be a $k$-soliton to the right of $r(\xi, 0)=0$. Let us check $D_{k}(\xi, x(\gamma^{1}))=k$.

Firstly, suppose $x(\gamma^{1})=x(\gamma)=a$. By Lemma~\ref{lem:xgamma}, this can only happen if $I(\gamma)\subset I(\gamma')$ for some $\gamma'$. Let $\tilde\gamma$ be the largest soliton among such $\gamma'$, which exists by virtue of Lemma~\ref{lem:basics}. Since $\gamma$ has to be appended to a $k$-slot, and since it can be checked there is no $k$-slot in $[x(\tilde\gamma), \cH_{k}(\tilde\gamma)]$, we deduce that $\cH_{1}(\tilde\gamma)< \dots<\cH_{k}(\tilde\gamma)<a$. Noting that $\T(\tilde\gamma^{1})$ is on the right of $\max I(\tilde\gamma)-1$, we see that $\cH_{1}(\tilde\gamma), \dots, \cH_{k}(\tilde\gamma)\in \mathcal D_{k}(\xi, a)$.
On the other hand, $\tilde\gamma$ is the only soliton in $ \xi $ containing $\gamma$ that is appended to a record, and therefore is also the only one that moves forward. If $\widehat\gamma$ is another soliton satisfying $x(\widehat\gamma)=\min I(\widehat\gamma)<a<\max I(\widehat\gamma^{1})$, then
we must have $I(\widehat\gamma)\cap I(\gamma)=\varnothing$. In particular, $\gamma$ is nested in the tail of $\widehat\gamma^{1}$, that is, $a>\T_{1}(\widehat\gamma^{1})$. Moreover, the first $k$-slot after $\T_{1}(\widehat\gamma^{1})$ is $\T_{k+1}(\widehat\gamma^{1})$. This implies $a\ge \T_{k+1}(\widehat\gamma^{1})$. We then conclude with $D_{k}(\xi, a)=\#\mathcal D_{k}(\xi, a)=k$.

Secondly, suppose $a'=x(\gamma^{1})>x(\gamma)$. Let us show $\mathcal D_{k}(\xi, a')=\{\cH_{1}(\gamma), \dots, \cH_{k}(\gamma)\}$ in this case. Indeed, if $\widehat\gamma\ne \gamma$ is a soliton with $x(\widehat\gamma)<a'<\max I(\widehat\gamma^{1})$, as $\gamma$ is appended to a record, we must have $I(\widehat\gamma)\cap I(\gamma)=\varnothing$ and $\widehat\gamma$ has size $>k$. As previously, we deduce that $\gamma^{1}$ is appended to a $k$-slot inside the tail of $\widehat\gamma^{1}$, so that $a'\ge \T_{k+1}(\widehat\gamma^{1})$. It follows $\cH_{i}(\widehat\gamma)\notin \mathcal D_{k}(\xi, a')$ for each $i\le k$, and therefore $D_{k}(\xi, a')=k$.

Now if $r(T\xi, 0)=r(\xi, 0)$, then Proposition~\ref{prop:solitontrack} and Lemma~\ref{lem:basics} imply that $o^{1}_{k}(\xi)=0$ for all $k$. In the case that $x(\gamma^{1})=x(\gamma)$, the first statement in the theorem readily follows. If, on the other hand, $x(\gamma^{1})>x(\gamma)$, denote by $\mathcal T_{1}$ the first site in the tail of $\gamma$; then according to Proposition~\ref{prop:solitontrack} and Lemma~\ref{lem:xgamma}, we must have $x(\gamma^{1})=\mathcal T_{1}$. However, $[x(\gamma), \mathcal T_{1}]$ can only contain solitons of size $\le k$ by Lemma~\ref{lem:basics}, from which it follows that $S_{k}\xi\cap [x(\gamma), x(\gamma^{1}))=\varnothing$. So the first statement in the theorem also holds in this case. 
If $r(T\xi, 0)<r(\xi, 0)$, then we need to incorporate a further shift into the indices of $k$-slots which is precisely given by $o^{1}_{k}(\xi)$. If $\gamma$ is to the left of $r(\xi, 0)$, the arguments are similar: it suffices to note that after each iteration, the soliton loses $k$ $k$-slots on the right.

To extend to an infinite configuration $\xi\in \cW$, let us take $x\in \Z$ and let $r(x)$ be the rightmost record of $\xi$ preceding $x$. We note that $T\xi|_{[r(x), x]}$ only depends on $\xi|_{[r(x), x]}$ in the sense that if $\xi'$ is a finite configuration with $\xi'(y)=\xi(y)$ for all $y\in [r(x), x]$ and $r(x)$ being a record of $\xi'$, then $T\xi(y)=T\xi'(y)$ for all $y\in [r(x), x]$. Since solitons are contained in excursions, we also have $S\xi(y)=S\xi'(y)$ for $y\in [r(x), x]$. Let $r'(x)=r(T\xi, -\xi(r(x)))\le r(x)$ be the record in $T\xi$ which replaces $r(x)$. Then the same argument implies that $ST\xi(x)$ only depends on $T\xi|_{[r'(x), x]}$, which in its turn only depends on $\xi|_{[r(r'(x)), x]}$.
Therefore, if we take a large enough box $[-n, n]$ containing $[r(r'(x)), x]$ as well as $r(r(T\xi, 0))$, then the restriction of $\xi$ to this box is enough to determine the truncated $k$-slot configurations $(s_{k}(\xi, i))_{i\le i_{0}}$ with $i_{0}=\max\{i: s_{k}(\xi, i)\le x\}$ and $(s_{k}(T\xi, i))_{i\le i'_{0}}$ with $i'_{0}=\max\{i: s_{k}(T\xi, i)\le x\}$. Hence, the first statement of the theorem also holds for a general $\xi\in\cW$.

For the third statement, suppose $\xi,\xi'\in\cW$ satisfy $(M_m\xi :m>k)=(M_m\xi':m>k)$; 
let us show that $J^{1}_{m}\xi=J^{1}_{m}\xi'$ for all $m> k$, from which it will follow $o^t_k(\xi)= o^t_k(\xi')$ by~\eqref{itk}. 
We have seen that there exists some $k_{0}(\xi)\in \N$ (resp.~$k_{0}(\xi')\in \N$) such that $J^{1}_{m}\xi=0$ for all $m\ge k_{0}(\xi)$ (resp.~$J^{1}_{m}\xi'=0$ for all $m\ge k_{0}(\xi')$). Taking $k_{0}=\max(k_{0}(\xi), k_{0}(\xi'))$, we see that $J^{1}_{m}\xi=J^{1}_{m}\xi'=0$ for all $m\ge k_{0}$. For $k=k_0-1$, we immediately deduce $o^1_k(\xi)=o^{1}_{k}(\xi')=0$ by~\eqref{itk}.  Moreover, we have $J^{1}_{k}\xi$ counting the number of $k$-solitons crossing Record $0$, namely, the $k$-solitons appended to a $k$-slot with a negative index at time $0$ and then appended to a positive-indexed $k$-slot at time $1$. Thus, 
$J^{1}_{k}\xi = \sum_{i=0}^{k-1}M_{k}T\xi(i)=\sum_{i=-k}^{-1}M_{k}\xi(i)$, by the second statement of the theorem. A similar identity holds for $J^{1}_{k}\xi'$. Therefore, if $k=k_{0}-2$, we have $J^{1}_{m}\xi=J^{1}_{m}\xi'$ for all $m>k$ and so  $o^t_k(\xi)= o^t_k(\xi')$. Proceeding by downward induction, the same will be true for $k=k_0-2,k_0-3,\dots,2,1$.
This concludes the proof of Theorem~\ref{thm:hierarchy}.
\end{proof}

\section{Invariant measures}
\label{sec:measures}

It is known that stationary Markov chains on $ \{0,1\} $ with density of 1's less than $\frac12$ are $T$-invariant,\footnote
{The operator $T$ applied to walks $\xi$ is equivalent to the $2M-X$ Pitman operator. Using reversibility and Burke arguments,~\cite{HamblyMartinOConnell01} show that product measures and stationary Markov chains with density less than $\frac12$ are $T$-invariant. Extensions and a discussion of the relation between BBS and the Pitman operator can be found in~\cite{CroydonKatoSasadaTsujimoto18,CroydonSasada19,croydon2020discrete}.}
but in fact there are many other invariant measures for the BBS. This is due to the existence of many conservation laws intrinsic to this dynamics, in particular the conservation of solitons studied in the previous section

In this section we prove Theorem~\ref{thm:invariant}.
More precisely, we show explicitly how invariant measures can be constructed by specifying the distribution of each $k$-th component $\zeta_k$.

We will refer to probability measures as simply \emph{measures}.
We will also refer to measurable functions as \emph{random elements}, and refer to the push-forward of a pre-specified measure
by such functions
as the \emph{law} of these random elements.
A measure $\mu$ on $\cX$ is \emph{$T$-invariant} if $\mu \circ T^{-1} = \mu$.

\subsection{Construction of the measures}
\label{sub:muconstruct}

Our recipe to produce invariant measures uses the construction described in \S\ref{sub:reconstruction}, and it gives a distribution $\hmu$ of configurations ``seen from a typical record.''
The proof of Theorem~\ref{thm:invariant} is based on properties of $\hmu$ and how they relate to the the dynamics.
This relationship is given by the Palm theory, which we briefly recall now.

Let $\eta\in \{0,1\}^\Z$.
From~\eqref{eq:reta} and~\eqref{def: theta}, we have $\theta R \eta = R \theta \eta$.
Suppose $\eta\in\cX$, and denote
\[
\ww(\eta) := \inf\{x\ge 1: x \in R\eta\}
.
\]
If $ \eta $ is random with law $ \hmu $
such that $\hmu(\ww):=\int \ww(\eta)\,\hmu(\dd\eta) <\infty$, its \emph{inverse-Palm measure} $\mu=\Palm_{R}^{\Z}(\hmu)$ is defined as follows.
For every test function $\varphi$,
\begin{align}
\label{eq:palm1}
\int \varphi(\eta) \,\mu(\dd\eta)
:=
\int
\frac
{\textstyle\sum_{i=0}^{\ww(\eta)-1} \varphi(\theta^i\eta)}
{\ww(\eta)}
\cdot
\frac{\ww(\eta)}{\hmu(\ww)}
\hmu(\dd\eta)
.
\end{align}
Informally, to sample a configuration distributed as $\mu$ one can first sample a configuration using the distribution $\hmu$ biased by the length of the first excursion (when $ 0\in R\eta $), and then choose a site uniformly from this excursion (along with the record preceding it) to place the origin.

For $ \eta \in \cX $,
we define the \emph{record-shift} operator $\htheta$ as
\[
\htheta \eta := \theta^{\ww(\eta)}\eta
\]

We also define the \emph{dynamics seen from a record} $\hT:\hcX\to\hcX$ by
\begin{align}
\label{eq:defht}
\hT \eta := \theta^{r(T \bxi[\eta],0)}T\eta,
\end{align}
where $\bxi[\eta]$ was defined in \S\ref{sub:reconstruction} as the unique lift of $\eta$ with Record $0$ at $0$.
In words, $ \hT $ means apply $ T $ and recenter the configuration at the new position of Record~0.
The following lemma follows from standard properties of Palm measures and is proved in \S\ref{sec:palm}

\begin{lemma}
\label{lemma:harris}
Let $\hmu$ be a probability measure of $\{0, 1\}^{\Z}$.
Suppose that $\hmu(0\in R\eta)=1$, $\hmu(\ww)< \infty$, $\hmu$ is $\htheta$-invariant and $\hmu(R\eta=\Z)=0$.
Then $\hmu$ is supported on $\hcX$, and $\Palm_{R}^{\Z}(\hmu)$ is $\theta$-invariant, supported on $\cX$ and satisfies
\[
\Palm_{R}^{\Z}(\hmu) \circ T^{-1} = \Palm_{R}^{\Z}(\hmu \circ \hT^{-1})
.
\]
If moreover $\hmu$ is $\htheta$-ergodic, then $\Palm_{R}^{\Z}(\hmu)$ is $\theta$-ergodic.
\end{lemma}

We now introduce the assumptions and notation used throughout the rest of this section.

Let $\zeta=(\zeta_k)_{k\ge 1}$ be a sequence of independent random elements of $(\Z_{+})^{\Z}$ with shift-invariant distributions.
Let $ P $ and $ E $ denote the underlying probability measure and expectation, and suppose the law of $ \zeta $ satisfies $\sum_k k E[\zeta_k(0)] <\infty$ and $P(\sum_{k,i}\zeta_{k}(i)>0)=1$.
In particular, this implies
the random field $\zeta$ a.s.\ satisfies~\eqref{eq:finitezeta} by Borel-Cantelli.
Take $\xi=M^{-1} \zeta$ as the walk reconstructed from $\zeta$ according to the algorithm described in \S\ref{sub:reconstruction} and depicted in Fig.~\ref{fig:zeta-2}.
Let $\hmu$ denote the resulting law of $\eta=\eta[\xi]$.

\begin{proposition}
\label{prop:invariant}
The measure $\hmu$ defined above is $\hT$-invariant and $\htheta$-invariant, and it also satisfies $\hmu(\ww)<\infty$ and $\hmu(R\eta=\Z)=0$.
If moreover $(\zeta_k(i))_{i\in\Z}$ is i.i.d.\ for each $k$, then $\hmu$ is also $\htheta$-ergodic.
\end{proposition}
Before giving the proof, let us see how it implies Theorem~\ref{thm:invariant}.

\begin{proof}
[Proof of Theorem~\ref{thm:invariant}]
Let $\hmu$ be the law of $\eta = M^{-1}\zeta$, and define
\[
\mu:=\Palm_{R}^{\Z}(\hmu)
.
\]
By Proposition~\ref{prop:invariant} and Lemma~\ref{lemma:harris}, $\mu$ is $\theta$-invariant and supported on $\cX$.
Also, $$\mu\circ T^{-1} = \Palm_{R}^{\Z}(\hmu\circ\hT^{-1})=\Palm_{R}^{\Z}(\hmu)=\mu,$$ so $\mu$ is also $T$-invariant.
Moreover, under the i.i.d.\ assumption, the second part of Proposition~\ref{prop:invariant} combined with the second part of Lemma~\ref{lemma:harris} implies that $\mu$ is $\theta$-ergodic.
\end{proof}

\begin{remark}
[i.i.d.\ measures]
A natural question is whether the product measures $\nu_\lambda$ can be constructed in this way.
This is indeed the case, as shown in~\cite{FerrariGabrielli20}.
\end{remark}

\begin{example}
[Ergodicity without independent components]
\label{exp1}
It is possible for law $\mu$ of $\eta$ to be $\theta$-ergodic and $T$-invariant while its components $M_k \eta$ not being independent under $\hmu$.
Let $\zeta'$ be the configuration $\zeta'(x) = \one\{x \ \mathrm{mod} \ 3 = 0\}$. Let $\zeta_1=\zeta_4$ be a configuration chosen uniformly at random in the set $\{\zeta', \theta\zeta', \theta^2\zeta'\}$; let $\zeta_k\equiv 0$ for all $k\notin\{1,4\}$ and $\zeta=(\zeta_k)_{k\ge 1}$.
The reader can check that this example satisfies the stated properties.
The resulting configuration is periodic for $ \theta $ and invariant for $ T $:
\[
\bf
\cdots
00{\blue 10}{\red 1111{\blue 01}0000{\blue 10}}000{\blue 10}{\red 1111{\blue 01}0000}{\blue 10}0
\cdots,
\]
where colors (online) represent records, $ 4 $-solitons and $ 1 $-solitons.
\end{example}

\begin{example}
[Independent components without ergodicity]
\label{exp2}
It is also possible for $\zeta$ to be independent over $k$, $\theta$-ergodic for each $k$, but produce (by the above procedure) a configuration $\eta$ whose law is not $\theta$-ergodic.
To see that, take $\zeta_5(x)\equiv 1$, $\zeta_1$ as in the previous example, $\zeta_k\equiv 0$ for all $k\notin\{1,5\}$ and $\zeta=(\zeta_k)_{k\ge 1}$.
The resulting configurations give three classes periodic for $ \theta $ and cyclic over $ T $:
\begin{gather*}
\bf
\cdots
0{\blue 10}{\red 1111{\blue 01}1000{\blue 10}00}%
0{\blue 10}{\red 1111{\blue 01}1000{\blue 10}00}%
\cdots
\\
\bf
\cdots
0{\red 111{\blue 01}1100{\blue 10}000{\blue 10}}%
0{\red 111{\blue 01}1100{\blue 10}000{\blue 10}}%
\cdots
\\
\bf
\cdots
0{\red 11{\blue 01}111{\blue 01}0000{\blue 10}0}%
0{\red 11{\blue 01}111{\blue 01}0000{\blue 10}0}%
\cdots
\end{gather*}
where colors (online) represent records, $ 5 $-solitons and $ 1 $-solitons.
\end{example}

\begin{remark}
\label{rmk:conjecture}
A measure $\mu$ is said to be $\theta$-mixing if for all $A, B \in \mathcal{B}$, $\mu(A\cap \theta^{-n}B) {\to} \mu(A)\mu(B)$ as $ n\to\infty $.
We conjecture that, for every measure $ \mu $ supported on $ \cX $,
$ \mu $ is $T$-invariant and $\theta$-mixing if and only if under $ \hmu $ the components $\zeta_k$ are independent over $k$ and each one is $\theta$-mixing.
\end{remark}

\subsection{Invariance of the reconstructed configuration}
\label{sub:invariance}

We now prove the main part of Proposition~\ref{prop:invariant}, namely $\htheta$-invariance and $\hT$-invariance of $\hmu$, as well as $\theta$-ergodicity in case of i.i.d.\ components.
The proof of $\hmu(\ww)<\infty$ is given in \S\ref{sub:finitew}.
The condition $\hmu(R\eta=\Z)=0$ is easily checked from the construction of $M^{-1}\zeta$ and the assumption $\sum_{k,i}\zeta_{k}(i)>0$ almost surely.
Denote by
$\cF_k$ the sigma-field generated by $(\zeta_m)_{m>k}$.

Since one can write $\eta = M^{-1}M\eta$ and $\hT \eta = M^{-1}M\hT\eta$ by Theorem~\ref{bijection1}, it suffices to show that the slot decompositions $M\eta$ and $M\hT\eta$ have the same law.
More precisely, it suffices to show 
\begin{align}
\label{ttt}
E\Bigl(\prod_{k=1}^n \varphi_k(M_k\hT\eta)\Bigr) = \prod_{k=1}^n E \varphi_k(\zeta_k)
\end{align}
for test functions $\varphi_1,\dots,\varphi_n$ and $n \in \N$.

We proceed by an induction on $n$.
First,
recall the notation from~\eqref{eq:defht} and
note that
\begin{align*}
&E\big(\varphi_k(M_k\hT\eta)\,\big|\, \cF_k \big)
\\
&\qquad= E\big(\varphi_k(\theta^{-o^1_k(\bxi[\eta])-k}M_k\eta)\,\big|\, \cF_k \big) \qquad (\hbox{by Theorem~\ref{thm:hierarchy}})
\\
&\qquad= E\big(\varphi_k(\theta^{-o^1_k(\bxi[\eta])-k}\zeta_k)\,\big|\, \cF_k \big) \qquad (\hbox{because }M_k\eta=\zeta_k)
\\
&\qquad= E\varphi_k(\zeta_k)
,
\end{align*}
because $\zeta_k$ is shift-invariant and independent of $(\zeta_m)_{m>k}$ whereas $o^1_k(\bxi[\eta])$ is determined by these elements.
The inductive step to show~\eqref{ttt} is then
\begin{align}
E\Bigl(\prod_{i=k}^n \varphi_i(M_i\hT\eta)\Bigr)
&= E\Bigl(E\Bigl(\prod_{i=k}^n \varphi_i(M_i\hT\eta)\Big| \cF_k \Bigr)\Bigr),
\\
&= E\Bigl( \prod_{i=k+1}^n \varphi_i(M_i\hT\eta) E\Bigl(\varphi_k(M_k\hT\eta)\Big| \cF_k \Bigr)\Bigr),
\\
&= E \varphi_k(\zeta_k) \ E\Bigl(\prod_{i=k+1}^{n} \varphi_i(M_i\hT\eta)\Bigr);
\end{align}
in the second identity we have used that $M_i\hT\eta$ is determined by $(\zeta_m)_{m\ge i}$ by the last statement of Theorem~\ref{thm:hierarchy}.
This shows that $\hmu$ is $\hT$-invariant.

We now prove $ \htheta $-invariance.
Consider the transformation $M^{-1} : \zeta \mapsto \xi^*$ defined in \S\ref{sub:reconstruction} and 
denote $\eta^{*}=\eta[\xi^{*}]$ the corresponding ball configuration. 
Call $\vep^0_*$ the excursion of $\eta^*$ between Records~0 and~1. The construction of \S\ref{sub:reconstruction} gives
\begin{align}
\htheta \eta^*
=
\theta^{r(\bxi[\eta^*],1)}\eta^*
=
M^{-1} \big(\theta^{n_k(\vep^0_*)}\zeta_k: k\ge 1\big).
\end{align}
So it suffices to show that $( \theta^{n_k(\vep^0_*)}\zeta_k )_{k\ge1}$ has the same law as $(\zeta_k)_{k \ge 1}$.
But $n_k(\vep^0_*)$ is determined by $(\zeta_m)_{m>k}$, thus independent of $\zeta_k$.
Hence the law of $\zeta_k$ is invariant by the random shift of $n_k(\vep^0_*)$ and it is independent of $(\zeta_m)_{m>k}$.
This shows that $ \eta^* $ and $ \htheta \eta^* $ have the same law.

Finally, under the extra assumption that $(\zeta_k(i))_{i\in\Z}$ is i.i.d.\ for each $k$,
the reconstruction map mentioned above will produce an i.i.d.\ sequence of excursions separated by records.
This in turn implies that the resulting configuration $\eta$ is $\htheta$-ergodic.

\subsection{Expected excursion length}
\label{sub:finitew}

We continue with the proof of Proposition~\ref{prop:invariant} to prove that $\hmu(\ww)<\infty$.
Let
\begin{equation}
\label{eq:alphak}
\alpha_k :=
E[\zeta_k(0)]
.
\end{equation}
The proof consists of two steps: first, we show that the following system
\begin{align}
\label{eq:skk1}
w_k &= 1 + {\sum_{m>k}} 2(m-k) w_m\alpha_m,\qquad k=0,1,2,\dots
\end{align}
has a unique finite non-negative solution $w=(w_k)_{k\ge 0}$, i.e., $ w \in [0,\infty)^{\N_0} $; second, we show that
the average number of $k$-slots per excursion in $M^{-1} \zeta$ is $w_k$, whence the average number of $k$-solitons per excursion satisfies
\begin{equation}
\label{eq:rhok}
\rho_k=\alpha_k w_k
.
\end{equation}
In particular, this will imply that the average size of the excursions (along with the record preceding them) satisfies
\begin{align}
\label{eq:wfinite}
\hmu(\ww) = w_0 = 1 + \sum_{m\ge1} 2m\, \rho_m <\infty.
\end{align}

So we start by studying~\eqref{eq:skk1}.
Let
\(
c_k := 2 {\sum_{m>k}}(m-k) \alpha_m.
\)
Since $\sum_{k} k E[\zeta_{k}(0)]<\infty$, we can take $\tilde{k}$ such that
\(
\sum_{m>\tilde{k}} 4 m \alpha_m < 1
,
\)
so
$c_k< \frac{1}{2}$ for $k\ge \tilde k$. Let $K:= \{k\in\N:k\ge \tilde k\}\cup \{\aleph\}$ and consider a Markov chain $(X_n)_{n\ge 0}$ on $K$ with absorbing state $\aleph$ and transition probabilities $q_{k,m} := 2(m-k) \alpha_m \one\{m>k\}$; $q_{k,\aleph} = 1-c_k$; $q_{\aleph,\aleph}=1$ and $q_{k,m}=0$ otherwise. Define the absorption time by
$\tau:= \inf\{n\ge 0: X_n=\aleph\}$.
Denote by $P_k$ the law of $(X_n)_{n\ge 0}$ starting from $k$, and by $ E_k $ the expectation.
Applying the Markov property at time $1$, we see that $w_k := E_k\tau$ satisfies
\begin{align}
\label{eq:markov}
w_{k}=E_{k}[\tau\one\{X_{1}=\aleph\}]+\sum_{m>k}E_{k}[\tau\one\{X_{1}=m\}] = (1-c_{k})+ \sum_{m>k}(1+w_{m})q_{k,m}.
\end{align}
That is, $(w_{k})_{k\ge \tilde k}$ verifies the system~\eqref{eq:skk1} with $k= \tilde k, \tilde k+1, \dots$.
Since $c_k\ge c_{k+1}$, we have $P_k(\tau>n) \le c_k^{n}$ and thus
\(
w_k= E_k\tau \le \frac{1}{1-c_k} < 2 < \infty, \hbox{ for $k\ge \tilde k$}.
\)
Since $w_{\tilde{k}}<\infty$, we can use~\eqref{eq:skk1} with $k=\tilde{k}-1$ to define $w_{\tilde{k}-1}<\infty$, and iterating this argument we get $w_k < \infty$ for all $k$.
This proves the existence of a finite solution to~\eqref{eq:skk1}.

For uniqueness, suppose $(\widetilde w_{k})_{k\ge 0}$ is a finite non-negative solution to~\eqref{eq:skk1}. In particular, $\sum_{m>k}\widetilde w_{m}q_{k,m}=\sum_{m>k}2(m-k)\widetilde w_{m}\alpha_{m}<\infty$, $k\ge 0$. Moreover, the previous choice of $\tilde k$ ensures that $\sum_{m>k}q_{k,m}<\infty$ for all $k\ge \tilde k$. It follows that for all $k\ge \tilde k$,
\[
|w_{k}-\widetilde w_{k}|
\le
 \sum_{m>k} q_{k,m} |w_{m}-\widetilde w_{m} \big|
\le
c_{k}\sup_{m>k} |w_{m}-\widetilde w_{m}|
\leq \tfrac{1}{2} \sup_{m>k} |w_{m}-\widetilde w_{m}|
.
\]
Taking the supremum over $k \geq \tilde{k}$ we get
$$\sup_{k\ge \tilde k}|w_{k}-\widetilde w_{k}| \leq \frac12 \sup_{k \geq \tilde{k}}|w_{k}-\widetilde w_{k}|.$$
By~\eqref{eq:skk1}, $\widetilde w_{k}$ and $ w_k $ are decreasing in $ k $, so the above supremum is finite, hence it is zero.
That is, $\widetilde w_{k}=E_{k}\tau$, $k\ge \tilde k$.
Given $(\widetilde w_{k})_{k\ge \tilde k}$,~\eqref{eq:skk1} determines in a unique way the values of $(\widetilde w_{k})_{k\le \tilde k}$; this completes the proof of uniqueness.

We now consider truncated approximations for the reconstruction algorithm of \S\ref{sub:reconstruction}.
Let
\begin{equation}
\label{eq:zetatrunc}
\zeta^{[n]}_k(i) :=
\begin{cases}
\zeta_k(i), & k \le n,
\\
0 , & k>n.
\end{cases}
\end{equation}
Let $\vep^{[n]}$ denote the first excursion (i.e.\ the one between Records~0 and~1) of $M^{-1} \zeta^{[n]}$.
Let $W^n_k\zeta = n_k(\vep^{[n]})$ be the number of $k$-slots in $\vep^{[n]}$; see~\eqref{sup-vep}.
Then $W^n_k\zeta \nearrow W_k\zeta$ a.s., where $W_k \zeta := W_k^\infty \zeta$.
Letting $w_k^n:=E[W^n_k\zeta]$ and $w_k:=E[W_k\zeta]$, by monotone convergence we have $w^n_k\nearrow w_k$ as $n\to\infty$.
On the other hand,
since each $m$-soliton contains $2(m-k)$ $k$-slots,
\begin{align}
W_k^n \zeta = 1 + \sum_{m>k} 2(m-k) \times (\hbox{number of $m$-solitons in $\vep^{[n]})$}
\end{align}
and thus
\begin{align}
w_k^n=1+\sum_{m>k} 2(m-k) E(\hbox{number of $m$-solitons in $\vep^{[n]})$}.
\end{align}
Let $\alpha^n_m := \alpha_m \one\{m\le n\}$ denote the expected number of $m$-solitons per $m$-slot in $\zeta^{[n]}$.
Since $W^n_k\zeta$ is a function of $(\zeta_m:m>k)$ which is independent of $\zeta_k$, the expected number of $m$-solitons in $\vep^{[n]}$ is $w_m^n \times \alpha_m^n$.
Therefore, $(w_k^n)_{k \ge 0}$ and $(\alpha_k^n)_{k \ge 1}$ satisfy relation~\eqref{eq:rhok} and the system~\eqref{eq:skk1}.
Finally, since $w^n_k < 2$ for all $k \ge \tilde{k}$ and $n\in\N$, $w_k$ is finite for every $k$ and therefore~\eqref{eq:wfinite} is satisfied, concluding the proof that $\hmu(\ww)<\infty$.

\section{Palm transformations}
\label{sec:palm}

We recall some fundamental properties on the Palm measures from Thorisson~\cite{Thorisson00}, and study their interplay with operator~$ T $ following~Harris~\cite{Harris71}.
Although Thorisson deals with processes in the continuum, adapting the arguments to discrete space is straightforward.

For our use in \S\ref{sec:speeds}, we will consider here a general situation of which $R\eta$ is a particular case.
Let $Z\eta$ be a subset of $\Z$ which depends on $\eta$ in a translation-covariant way ($Z\theta\eta = \theta Z\eta$).

The map $ \Palm_\Z^Z : \mu \mapsto \tmu $ is defined as follows.
Let $\mu$ be a $\theta$-invariant probability measure on $\{0, 1\}^{\Z}$ and assume that
$ \mu(Z\eta = \varnothing) = 0. $
Define
$$\tilde{\ww}=\tilde{\ww}(\eta) := \inf\{x\ge 1: x\in Z\eta\}$$
and $\ttheta := \theta^{\tilde{\ww}}$, the shift to the next element in $Z\eta$.

We define the measure $\tmu$ as follows: for each $ m \in \N $ and each test function $\varphi$,
\begin{align}
\label{def:Palm}
\int \varphi(\eta)\tmu(\dd\eta) :=
\frac
{\int \sum_{i=0}^{m-1} \varphi(\theta^{i}\eta) \one\{0 \in Z\theta^i \eta\} \mu(\dd \eta)}
{\int \sum_{i=0}^{m-1} \one\{0 \in Z\theta^i \eta\} \mu(\dd \eta)}
.
\end{align}

The above definition does not depend on $m$ (Theorem 8.3.1 of~\cite{Thorisson00}). In particular, if we specialize it to the case $m=1$, it becomes
\begin{align}
\label{eq:palm3}
\int \varphi(\eta)\tmu(\dd\eta) :=
\frac
{\int \varphi(\eta) \one\{0 \in Z\eta\} \mu(\dd \eta)}
{\int \one\{0 \in Z\eta\} \mu(\dd \eta)},
\end{align}
which means
\begin{equation}
\label{eq:palm2}
\tmu = \mu( \ \cdot\ |\, 0\in Z\eta).
\end{equation}

Theorem 8.4.1 and Formula~(8.4.14) of~\cite{Thorisson00} then assert that the measure $\tmu$ is $\ttheta$-invariant.
Moreover, by Formula~(8.4.6) of the same book we have that the mean distance between successive points of $Z\eta$ under $\tmu$ is the inverse density of $Z\eta$ under $\mu$:
\[
\tmu(\ww)=
\frac{1}{\mu(0\in Z\eta)}
< \infty
.
\]
Conversely, suppose $\tmu$ is a $\ttheta$-invariant measure on $ \{0,1\}^\Z $ satisfying
$\tmu(\tilde{\ww})<\infty$.
Then its inverse Palm measure $\mu = \Palm_Z^\Z \tmu $ is defined as follows: for every test function $\varphi$,
\begin{equation}
\label{eq:palm-gen}
\int \varphi(\eta) \,\mu(\dd\eta)
:=
\frac
{\int \textstyle\sum_{i=0}^{\tilde{\ww}(\eta)-1} \varphi(\theta^i\eta) \tmu(\dd\eta)}
{\int \textstyle\sum_{i=0}^{\tilde{\ww}(\eta)-1} 1 \tmu(\dd\eta)}.
\end{equation}
Moreover, $\mu$ is $\theta$-invariant (Theorem 8.4.1 and Formula (8.4.14$^{\circ}$) of~\cite{Thorisson00}), $ \mu(Z\eta = \varnothing)=0 $, and its Palm measure $\Palm_{\Z}^Z \mu$ is given by $\tmu$.

The above observations give the following.
\begin{lemma}
\label{lemma:bijection}
The operations~\eqref{eq:palm2} and~\eqref{eq:palm-gen} define a bijection between $\theta$-invariant measures $\mu$ on $ \{0,1\}^\Z $ with $ \mu(Z\eta = \varnothing) = 0 $ and $\ttheta$-invariant measures $\tmu$ on $ \{0,1\}^\Z $ with $\tmu(\tilde{\ww}) < \infty$.
\end{lemma}

We now analyze how the Palm transform relates to almost-sure properties.

\begin{lemma}
\label{lemma:aspalm}
Let $\mu$ be a $ \theta $-invariant measure on $ \{0,1\}^\Z $ with $ \mu(Z\eta = \varnothing) = 0 $, and $ A $ be a $ \theta $-invariant event.
Then $ \mu(A)=0 $ if and only if $ \tmu(A) = 0 $.
\end{lemma}
\begin{proof}
If $ \mu(A)=0 $, then $ \tmu(A)=0 $ by~\eqref{eq:palm2}.
Now suppose $ \tmu(A)=0 $.
Then, again by~\eqref{eq:palm2} we have $ \mu(A \cap \{0\in Z\eta\}) = 0 $.
By $ \theta $-invariance of both $ \mu $ and $ A $, and $ \theta $-covariance of $ Z $, this gives $ \mu(A \cap \{x\in Z\eta\}) = 0 $ for every $ x \in \Z $.
Taking union over $ x $, this gives $ \mu(A) \leq \mu(Z\eta = \varnothing) = 0 $.
\end{proof}

As a side remark, the denominator in~\eqref{eq:palm-gen} is simply $ \tmu(\tilde{\ww}) $, and this is the same formula as~\eqref{eq:palm1} with $ Z $ instead of $ R $.
We wrote it in this apparently cumbersome way to highlight the similarity with~\eqref{eq:palm3}.
The number $ 1 $ in the denominator of~\eqref{eq:palm-gen} equals the indicator that $ 0 \in \Z $ and the absence of a sum over $ i $ in~\eqref{eq:palm3} is due to the fact that the analog to $ \ww $ is just the distance from $ 0 $ to the next point in $ \Z $, which is $ 1 $.
This also helps explain~\eqref{eq:invert} below.

\begin{lemma}
\label{lemma:ergodic}
For $ \theta $-invariant $ \mu $ with $ \mu(Z\eta=\varnothing)=0 $,
$\mu$ is $\theta$-ergodic if and only if $\tmu$ is $\ttheta$-ergodic.
\end{lemma}
\begin{proof}
Suppose $\tmu$ is $\ttheta$-ergodic. Let us prove the $\theta$-ergodicity of $\mu$.
It suffices to show that the Cesàro limits on test functions are $\mu$-a.s.\ constant.
Let $ \eta\in\{0,1\}^\Z $ be such that $ Z\eta $ is bi-infinite.
Write $k_n=\#Z\eta\cap(0, n)$.
For a non-negative bounded test function $ \varphi $, we rewrite its partial average as
\begin{align}
\frac{\sum_{x=0}^{n-1} \varphi (\theta^x \eta)}{n}
&=
\frac{\frac{1}{k_n}}{\frac{1}{k_n}}
\cdot
\frac
{\sum_{i=1}^{k_n-1}\sum_{j=0}^{\tilde{\ww}(\ttheta^{i}\eta)-1}\varphi(\theta^j \ttheta^i \eta) + \sum_{x=0}^{k_1-1}\varphi(\theta^x\eta) + \sum_{x=k_n}^{n-1}\varphi(\theta^x\eta) }
{\qquad \sum_{i=1}^{k_n-1} \tilde{\ww}(\ttheta^{i}\eta) \qquad + \qquad k_1 \qquad + \quad (n-k_n)}.
\end{align}
Applying the Ergodic Theorem~\cite[Chapter~2]{Coudene16} to $\tmu$ and $\ttheta$, we find that the event $ A $ given by the set of all configurations $ \eta $ such that
\[
\lim_{n\to\infty} \frac1n\sum_{x=0}^{n-1} \varphi (\theta^x \eta) = \frac{\int\textstyle\sum_{j=1}^{\tilde{\ww}(\eta)} \varphi(\theta^j\eta) \, \tmu(\dd\eta)}{\int\textstyle \tilde{\ww}(\eta) \, \tmu(\dd\eta) }
\]
satisfies $ \tmu(A)=1 $ (because the second and third terms of the numerator are bounded by the summands with $ i=0 $ and $ i=k_n $, and the same holds for the denominator).
By Lemma~\ref{lemma:aspalm}, this proves that $\mu$ is $\theta$-ergodic.

The implication in the other direction is proved similarly. A brief sketch is
\[
\frac{\sum_{i=0}^{k-1}\varphi(\ttheta^i \eta)}{k}
=
\frac{\frac{1}{n_k}\sum_{x=0}^{n_k-1}\bar\varphi(\ttheta^x \eta)}{\frac{k}{n_k}} \xrightarrow[k\to\infty]{\mu-\text{a.s.}} \frac{\int \bar\varphi\,\dd\mu}{\mu(0\in Z\eta)},
\]
where $ \bar\varphi(\eta):=\varphi(\eta)\one\{0\in Z\eta\} $ and $ n_k $ is the position of the $ k $-th element of $ Z\eta $.
\end{proof}

We now describe a rather useful consequence of the above theory, to be used in \S\ref{sec:speeds}.
It is given by the following diagram:
\begin{displaymath}
\xymatrixcolsep{7pc}
\xymatrixrowsep{5pc}
\xymatrix{
& \mu
\ar@/_.5pc/@{->}[dl]_{\Palm_\Z^R}
\ar@/^.5pc/@{->}[dr]^{\Palm_\Z^Z}
\\
\hmu
\ar@/_.5pc/@{->}[ur]_<(.6){\Palm_R^\Z}
\ar@/^.5pc/@{->}[rr]^{\Palm_R^Z}
&&
\tmu
\ar@/^.5pc/@{->}[ul]^<(.6){\Palm_Z^\Z}
\ar@/^.5pc/@{->}[ll]^{\Palm_Z^R}
}
\end{displaymath}
where $ \mu $, $ \hmu $ and $ \tmu $ describe respectively the configuration seen from a typical site, a typical record, or a typical element of $ Z $.
More precisely, suppose $ \mu(Z\eta = \varnothing) = \mu(R\eta = \varnothing) = 0 $.
Then by Lemma~\ref{lemma:aspalm} the same holds for $ \tmu $ and $ \hmu $.
In this case, $ \tmu $ and $ \hmu $ are related by
\begin{equation}
\label{eq:invert}
\int \varphi(\eta) \tmu(\dd\eta)
=
\frac
{
\int
\textstyle
\sum_{i=0}^{\ww(\eta)-1}
\varphi(\theta^i\eta)
\one\{0 \in Z\theta^i\eta\}
\,
\hmu(\dd\eta)
}
{
\int
\textstyle
\sum_{i=0}^{\ww(\eta)-1}
\one\{0 \in Z\theta^i\eta\}
\,
\hmu(\dd\eta)
}
.
\end{equation}
For the inverse transform $ \Palm_Z^R $, the same formula is valid if we swap $ \tmu,Z,\tilde{\ww}$ with $\hmu,R,{\ww} $.
Both~\eqref{eq:invert} and its inverse can be proved using~\eqref{eq:palm-gen} and~\eqref{eq:palm3}.

We finally move to studying the interplay between a Palm measure and the operator $ T $.\footnote{For concreteness, we consider the BBS operator $ T $, but we only use the fact that it is a $ \theta $-covariant operator and its domain $ \cX $ is $ \theta $-invariant.}

Let $ \cZ := \{ \eta\in\cX : Z\eta \text{ is bi-infinite}\} $.
Suppose for each $ \eta\in\cZ $ there is a bijection $\Psi_{\eta}$ between $Z\eta$ and $ZT\eta$ that depends on $ \eta $ in a translation covariant way, that is, $ \Psi_{\theta\eta}(\theta x) = \theta \Psi_\eta(x) $.
Intuitively, $ \Psi $ is just an honest way to follow elements of $ Z\eta $ after applying the operator $ T $.
Let
\[
\tcZ:=\{ \eta\in\cZ : 0\in Z\eta \}
.
\]
For $\eta\in \tcZ$, we define
$$
\widetilde T\eta=\theta^{\Psi_{\eta}(0)}T\eta
.
$$
We now show the following.

\begin{lemma}
\label{lemma:palmconjugate}
For $ \theta $-invariant $ \mu $ on $\cZ$, we have the identity
$$\Palm_{\Z}^Z (\mu\circ T^{-1})= (\Palm_{\Z}^Z \mu) \circ \widetilde T^{-1}.$$
In particular, $\mu$ is $T$-invariant if and only if $\Palm_{\Z}^Z \mu$ is $\widetilde T$-invariant.
\end{lemma}

\begin{proof}
We follow the classical arguments in~\cite{Harris71,PortStone73}.
For a test function $\varphi$, we have
\begin{align}
\int \one\{0\in Z T\eta\}\, \varphi(T\eta) \, \mu(\dd \eta)
&=
\int
\sum_x
\one{\{x \in Z\eta,\Psi_\eta(x)=0\}} \, \varphi(T\eta) \, \mu(\dd \eta)
\\
&=
\sum_x
\int
\one{\{0 \in Z \theta^x \eta,\Psi_{\theta^x\eta}(0)=-x\}} \, \varphi(T\eta) \, \mu(\dd \eta)
\\
&=
\sum_x
\int
\one{\{0 \in Z \eta,\Psi_{\eta}(0)=-x\}} \, \varphi(T \theta^{-x} \eta) \, \mu(\dd \eta)
\\
&=
\int
\one\{0 \in Z\eta\}\,
\sum_x
\one{\{\Psi_{\eta}(0)=x\}} \, \varphi(T \theta^{x} \eta) \, \mu(\dd \eta)
\\
&=
\int
\one\{0 \in Z\eta\}\,
     \varphi(\tT \eta) \, \mu(\dd \eta)
.
\end{align}
The third identity holds by the translation invariance of $\mu$.
Taking $\varphi\equiv 1$ we get $(\mu\circ T^{-1})(\tcZ)=\mu(\tcZ)$.
Hence,
\begin{align}
  \Palm_{\Z}^{Z}(\mu\circ T^{-1})\,\varphi
&=
\tfrac{1}{(\mu\circ T^{-1})(\tcZ)}
\int \one\{0\in Z T\eta\}\, \varphi(T\eta) \, \mu(\dd \eta)
\\
&=
\tfrac{1}{\mu(\tcZ)}
\int
\one\{0 \in Z\eta\}\,
\varphi(\tT \eta) \, \mu(\dd \eta) = \big( (\Palm_{\Z}^Z\mu) \circ \tT^{-1}\big)\varphi. \qedhere
\end{align}
\end{proof}

\begin{proof}[Proof of Lemma~\ref{lemma:harris}]
Define a bijection $\psi_\eta$ between $R\eta$ and $RT\eta$ as follows.
For $x \in R\eta$, let $\psi_\eta(x) := r(T \xi,j)$, where $ \xi = \xi[\eta] $ and $x=r(\xi,j)$.
Note that this definition does not depend on the lift $\xi[\eta]$.
Note also that $ \psi_{\theta\eta}(\theta x) = \theta \psi_\eta(x) $.
Finally note that~\eqref{eq:defht} can be written as
\begin{align*}
\hT \eta := \theta^{\psi_\eta(0)}T\eta.
\end{align*}
All the statements in the lemma follow from the previous lemmas by taking $Z\eta=R\eta$ and $ \Psi=\psi $, except that $ \mu $ is supported on $ \cX $ and $ \hmu $ is supported on $ \hcX $.
By Lemma~\ref{lemma:aspalm}, it is enough to show the latter.

By the Ergodic Theorem, the limit
\[
\lambda(\eta)=
\lim_{y\to\infty}\,\frac1{y} \sum_{x=-y}^{-1}\eta(x) = \lim_{y\to\infty}\,\frac1{y} \sum_{x=1}^y\eta(x)
\]
exists for $ \mu $-a.e.\ $ \eta $, hence for $ \hmu $-a.e.\ $ \eta $, and we have to show that $ 0 < \lambda(\eta) < \frac{1}{2} $ for $ \hmu $-a.e.\ $ \eta $.

By the Ergodic Decomposition Theorem~\cite[Chapter~14]{Coudene16}, we can assume that $ \hmu $ is $ \htheta $-ergodic.
In the above limit, we can take a subsequence $y_{k}\to\infty$ with elements from $R\eta$, giving
\[
\lambda(\eta)
=
\lim_{k\to\infty} \frac{\sum_{x=1}^{y_k} \eta(x)}{y_k}
=
\lim_{k\to\infty}
\frac
{\sum_{j=0}^{k-1} \frac{1}{2}(\ww(\htheta^j \eta)-1)}
{\sum_{j=0}^{k-1} \ww(\htheta^j \eta)}
=
\frac
{\frac{1}{2}(\hmu(\ww)-1)}
{ \hmu(\ww)}
<
\frac{1}{2}
.
\]
By assumption, $ \hmu(R\eta=\Z) = 0 $, thus $ \hmu(\ww=1)<1 $, hence $ \hmu(\ww)>1 $ and the above limit is also positive, concluding the proof.
\end{proof}

\section{Asymptotic speed of solitons}
\label{sec:speeds}

In this section we prove Theorems~\ref{thm:hspeed} and~\ref{thm:speedsexplicit}
from a combination of simpler statements.
By looking at the dynamics as seen from a $k$-soliton, we show that
the soliton speed $v_k$ exists and equals the expected length of the jump of a typical $k$-soliton in one step.

In \S\ref{sub:speedexists} we show that the soliton speed $v_k$ that appears in~\eqref{eq:speedexists} is $\mu$-a.s.\ well-defined and is given by an explicit formula~\eqref{eq:vfrompalm}.
We also show that it is given another formula~\eqref{eq:anothervfrompalm}, and that it is finite.
Analyzing the interaction between solitons of different sizes, in \S\ref{sub:speedequation} we show that the soliton speeds $(v_k)_k$ satisfy~\eqref{eq:hspeedseq}.
For completeness, in \S\ref{sub:speedpositive}
we show that $ v_k $ is positive and increasing in $ k $ without assuming that $ \orho_k>0 $.
In \S\ref{sub:speedexplicit} we analyze the formula~\eqref{eq:anothervfrompalm} using the description of $\mu$ from \S\ref{sub:muconstruct} to show that the soliton speeds are given by~\eqref{eq:hspeedexplicit}.
Finally, in \S\ref{sub:vertical} we briefly mention the results about vertical speeds.

\subsection{Existence of speeds via Palm measure and ergodicity}
\label{sub:speedexists}

Here we show that $\lim_t \frac{1}{t}x(\gamma^t)$ exists $ \mu $-a.s.\ and we give an explicit formula for it.
Recall the definition of $x(\gamma)$ given before the statement of Theorem~\ref{thm:hspeed}.
Let $\Gamma_k^\circ\eta:=\{x(\gamma): \gamma\in\Gamma_k\eta\}$ denote the set of leftmost sites of $k$-solitons of $\eta$.
If $\rho_{k}=0$, then $\Gamma_{k}$ is empty and~\eqref{eq:speedexists} holds for all $ \gamma\in\Gamma_k(\eta) $ by vacuity, for any value of $ v_k $.

From now on, we assume $\rho_{k}>0$, which implies $\mu(0\in \Gamma_k^\circ\eta)>0$.
Since $\Gamma_{k}^\circ\theta\eta=\theta\Gamma_{k}^\circ\eta$, the Palm construction described in \S\ref{sec:palm} applies.

For $\gamma\in\Gamma_k\eta$ and $z=x(\gamma)$, we define $\Delta_\eta^k(z) := x(\gamma^1)-x(\gamma)$, the size of the jump of $k$-soliton $\gamma$ after one iteration of $T$.
For $z\not\in \Gamma_k^\circ\eta$ we set $\Delta_\eta^k(z)=0$.
With this notation, the displacement of a tagged $k$-soliton after $t+1$ iterations of $T$ can decomposed as
\[
x(\gamma^{t+1})-x(\gamma) = \Delta^k_\eta(x(\gamma)) + \Delta^k_{T\eta}(x(\gamma^1)) + \cdots + \Delta^k_{T^t\eta}(x(\gamma^t))
.
\]
We want to divide both sides by $t$ and use the Ergodic Theorem.
This will require a couple of subtle observations.
The first step is to consider the system as seen from a typical $k$-soliton.

Let $\hcX^k$ be the set of configurations in $\cX$ such that $\Gamma_k^\circ\eta$ is doubly-infinite and contain $ 0 $.
Let $\hmu_k:=\Palm_{\Z}^{\Gamma_k^\circ}(\mu)$ be the Palm measure of $\mu$ with respect to $\Gamma_k^\circ\eta \subseteq \Z$, i.e., $\hmu_k=\mu(\,\cdot\,|\,0\in \Gamma_k^\circ\eta)$.
For $\eta\in\hcX^k$, let $\ww^k(\eta) := \inf\{x\ge 1: x\in \Gamma_k^\circ\eta \}$ and define $\htheta_k : \hcX^k \to \hcX^k$ as the ``shift to the next $k$-soliton'' given by $\htheta_k \eta := \theta^{\ww^k(\eta)}\eta$.
Also, for a $k$-soliton $\gamma$ such that $x(\gamma)=0$, let $\hT_k\eta := \theta^{x(\gamma^1)} T\eta$ denote the dynamics as seen from a tagged $k$-soliton.

From Lemma~\ref{lemma:palmconjugate}, $ \hmu_k $ is $\hT_k$-invariant.
Now for $\eta \in \hcX^k$ and $\gamma$ with $x(\gamma)=0$, the above decomposition becomes
\[
x(\gamma^{t+1}) = \Delta^k_{\eta}(0) + \Delta^k_{\hT_k\eta}(0) + \Delta^k_{\hT_k^2\eta}(0) + \cdots + \Delta^k_{\hT_k^t\eta}(0).
\]
By the Ergodic Theorem for $(\hT_k^t\eta)_{t \in \Z}$, $\lim\limits_{t\to\infty} \frac{1}{t}x(\gamma^t)$ exists $\hmu_k$-a.s., and on average it equals
\begin{equation}
\label{eq:vfrompalm}
v_k
:=
\int \Delta_\eta^k(0) \, \hmu_k(\dd\eta)
\end{equation}
(to avoid integrability issues, note that by Lemma~\ref{lem:xgamma} we have $ \Delta_\eta^k(0) \geq 0 $).

It remains to show that this limit is in fact non-random.

Consider the field
\[
\tv_k(\eta,z) :=
\begin{cases}
\lim\limits_{t\to\infty}
\dfrac{x(\gamma^t)}{t}, &
\text{if } z = x(\gamma) \text{ for some } \gamma\in\Gamma_k(\eta), \\
0,&\hbox{otherwise.}
\end{cases}
\]
By Proposition~\ref{prop:solitontrack} and Lemma~\ref{lem:basics}, if $\gamma, \tilde\gamma$ are two $k$-solitons with $x(\gamma)\le x(\tilde\gamma)$, then $x(\gamma^{t})\le x(\tilde\gamma^{t})$ for all $t$.
Hence, we have
\begin{equation}
\label{eq:speednondec}
\tv_k(\eta,x) \le \tv_k(\eta,y) \text{ for all } x\le y \text{ in } \Gamma_k^\circ\eta
.
\end{equation}
On the other hand, by Lemma~\ref{lemma:ergodic} the measure $ \hmu_k $ is $\htheta_k$-ergodic.
So, under $ \hmu_k $ the sequence $ \big( \tv_k(\htheta^j_k \eta,0) \big)_{j \in \Z} \in [0,+\infty]^\Z $ is a.s.\ non-decreasing by~\eqref{eq:speednondec} and its law is $ \theta $-ergodic, which implies that it is $ \hmu_k $-a.s.\ equal to some constant $ v_k \geq 0 $.
So $ \tv_k (\eta,x) = v_k $ for all $ x \in \Gamma^\circ_k(\eta) $, for $ \hmu_k $-a.e.\ $ \eta $.
By Lemma~\ref{lemma:aspalm}, it also holds for $ \mu $-a.e.\ $ \eta $.

We conclude with a short proof that $ v_k<\infty $.
From~\eqref{eq:invert} with $\Gamma_k^\circ$ in the role of $ Z $, we get
\begin{equation}
\label{eq:anothervfrompalm}
v_k
=
\frac
{
\int
\textstyle
\sum_{i=0}^{\ww(\eta)-1}
\one\{\theta^i\eta\in\hcX^k\}
\,
\Delta_{\theta^i\eta}^k(0)
\,
\hmu(\dd\eta)
}
{
\int
\textstyle
\sum_{i=0}^{\ww(\eta)-1}
\one\{\theta^i\eta\in\hcX^k\}
\,
\hmu(\dd\eta)
}
=
\frac{1}{\rho_k}
\int
\sum_{y=0}^{\ww(\eta)-1}
\Delta_{\eta}^k(y)
\,
\hmu(\dd\eta)
.
\end{equation}
Denote by $h(\gamma)$ the leftmost site of the tail of $\gamma$ and by $t(\gamma)$ the leftmost site of its head. Then Proposition~\ref{prop:solitontrack} implies that $\Delta_{\eta}^{k}(x)\le |h(\gamma)-t(\gamma)|$. Combined with the fact that intervals spanned by different $k$-solitons do not overlap (Lemma~\ref{lem:basics}), this yields $\sum_{y=1}^{\ww(\eta)} \Delta_{\eta}^k(y) \le \ww(\eta)$,
from which we get $v_k \le \frac{\hmu(\ww)}{\rho_k} = \frac{w_0}{\rho_k} < \infty$ by~\eqref{eq:wfinite}.

\subsection{Equation for speeds from soliton interactions}
\label{sub:speedequation}

We now prove that the soliton speeds satisfy~\eqref{eq:hspeedseq}.
It is enough to consider $ k $ such that $ \rho_k>0 $, and otherwise take~\eqref{eq:hspeedseq} as the definition of $ v_k $.
We start by proving the following identity on the displacement of solitons.

\begin{proposition}
\label{prop:displacement}
We have
\begin{equation}
\label{eq:displacement}
x(\gamma^{t})-x(\gamma) = k\, t - 2k \sum_{m>k} \tfrac12 M_t^\gamma(m) + \sum_{m<k} 2m \, N_t^\gamma(m),
\end{equation}
where
\begin{align*}
M_{t}^{\gamma}(m)&=\#\{(\tilde\gamma, s): \tilde\gamma\in \Gamma_{m}\eta, 0\le s < t, I(\gamma^{s})\subset I(\tilde\gamma^{s}) \text{ and } x(\tilde\gamma^{s+1}) \ne x (\tilde\gamma^{s}) \}, \\
N_{t}^{\gamma}(m)&=\#\{(\tilde\gamma, s): \tilde\gamma\in \Gamma_{m}\eta, 0\le s<t, I(\tilde\gamma^{s})\subset I(\gamma^{s}) \text{ and } I(\tilde\gamma^{s+1})\cap I(\gamma^{s+1})=\varnothing\}.
\end{align*}
\end{proposition}

\begin{proof}
Before proving it, let us first explain the intuition underpinning~\eqref{eq:displacement}.
From Lemma~\ref{lem:xgamma}, we know that
as long as there exists some $\tilde\gamma$ satisfying $I(\gamma)\subset I(\tilde\gamma)$,
$\gamma$ is prevented from moving.
We will see that it takes two time steps for $\tilde\gamma$ to overtake $\gamma$: the first step occurs when $\gamma$ is nested in the right half of $\tilde{\gamma}$, after which it will be nested in the left half of $\tilde{\gamma}$; the second step occurs when $\gamma$ is nested in the left half of $\tilde{\gamma}$ after which it will no longer be nested inside $\tilde{\gamma}$. Since $\tilde\gamma$ can be inside an even larger soliton, during which it stays frozen, these two times are not necessarily consecutive. On the other hand, when $\gamma$ overtakes a smaller soliton of size $\ell$ say, it causes $\gamma$ to move $2\ell$ units more than it normally would have. Put formally, we have the following identity
\begin{equation}
\label{eq:collision}
x(\gamma^{1})-x(\gamma)=k-k\cdot\one\{\exists\,\tilde\gamma: I(\gamma)\subset I(\tilde\gamma) \text{ and } x(\tilde\gamma^{1}) \ne x (\tilde\gamma) \}+\sum_{\ell<k} 2\ell N^{\gamma}_{0}(\ell).
\end{equation}
The proof of this identity is elementary and is postponed to \S\ref{sec:postponed}. Proceeding with the proof of~\eqref{eq:displacement}, we note that if $\tilde\gamma$ exists in~\eqref{eq:collision}, then it has to be unique:
solitons $\tilde\gamma$ satisfying $I(\tilde\gamma)\supseteq I(\gamma)$ are ordered by their sizes and only the largest one will move forward.
Iterating the above $t$ times then yields~\eqref{eq:displacement}.
\end{proof}

Now take $\tilde\gamma$ to be the leftmost $m$-soliton in $[1, \infty)$ and $\gamma$ the rightmost $k$-soliton in $(-\infty, 0]$.
We have shown previously $v_{m}=\lim_{t\to\infty} \frac{1}{t}x(\tilde\gamma^{t})$ and $v_{k}=\lim_{t\to\infty} \frac{1}{t}x(\gamma^{t})$, $\mu$-a.s. Combined with Lemma~\ref{lemma:comparex}, this yields
\begin{equation}
\label{eq:vk-vm}
v_{k}\le v_{m} \quad \text{if} \quad k<m \text{ and } \rho_{k}, \rho_{m}>0.
\end{equation}

From now on, we will work under $\hmu_{k}$ and always take $\gamma$ to be the soliton with $x(\gamma)=0$. To lighten the notation, we will drop the superscript ${\gamma}$ from $ M_t $ and $ N_t $.
Recall that $\orho_m = \frac{\rho_m}{w_0}$ is the mean number of $m$-solitons per unit space.

Let us assume for the moment $\orho_m=0$ for all large $m$.

By~\eqref{eq:speedexists} and~\eqref{eq:displacement}, to get~\eqref{eq:hspeedseq} it suffices to show the following.
For $m>k$ and $\ell<k$:
\begin{equation}
\label{eq:collisions}
\displaystyle \tfrac{1}{2t} M_t(m) \to \orho_m (v_m-v_k) \quad\text{and}\quad \tfrac{1}{t} N_t(\ell) \to \orho_\ell (v_k-v_\ell) \quad\text{ in probability}.
\end{equation}

\begin{proof}[Proof of~\eqref{eq:collisions}]
Let us denote $\mathbf{M}_{t}(m)=\{(\tilde\gamma, s): \tilde\gamma\in \Gamma_{m}\eta, 0\le s< t, I(\gamma^{s})\subset I(\tilde\gamma^{s}) \text{ and } x(\tilde\gamma^{s+1}) \ne x(\tilde\gamma^{s}) \}$, the set of solitons which overtake $\gamma$, so that $M_{t}(m)=\#\mathbf{M}_{t}(m)$. Let $\gamma_{-}$ be the leftmost $m$-soliton satisfying $x(\gamma_{-}) > x(\gamma)$ and let $\gamma_{+}$ be the rightmost $m$-soliton satisfying $x(\gamma_{+}) < x(\gamma)$.
Let us consider the $m$-solitons $\tilde\gamma$ which are found between $\gamma^{t}$ and $\gamma_{-}^{t}$ at time $t$ and denote by $\mathbf{L}^{-}_{t}$ their set, namely,
$\mathbf{L}^{-}_{t} := \{ \tilde\gamma\in \Gamma_{m}\eta: x(\gamma^{t}) < x(\tilde\gamma^{t}) < x(\gamma_{-}^{t}) \}$. 
Note that if $\tilde\gamma\in \mathbf L^{-}_{t}$, then necessarily $I(\gamma^{t})\cap I(\tilde\gamma^{t})=\varnothing$. 
We also define $\mathbf{L}^{+}_{t}=\{\tilde\gamma\in \Gamma_{m}\eta: I(\tilde\gamma^{t})\cap [x(\gamma^{t}), x(\gamma_{+}^{t})]\ne\varnothing\}$, 
which includes in particular the $m$-soliton $\tilde\gamma$ satisfying $I(\gamma^{t})\subset I(\tilde\gamma^{t})$ (if such soliton exists). 
Roughly speaking, $\#\mathbf{L}^{-}_{t}$ counts those $m$-solitons which have completed the 2-step overtaking of $\gamma$ between times $0$ and $t-1$, so it will give an undercount of $M_{t}(m)$, while $\#\mathbf{L}^{+}_{t}$ will be an overcount. 
More precisely, let us show that $2\#\mathbf{L}^{-}_{t}\le M_{t}(m)\le 2 \#\mathbf{L}_{t}^{+}$.
Indeed, if $(\tilde\gamma, s)\in \mathbf{M}_{t}(m)$, then Lemma~\ref{lemma:comparex} implies that necessarily $x(\tilde\gamma) < x(\gamma)$. It follows $x(\tilde\gamma)\le x(\gamma_{+})$, and then $x(\tilde\gamma^{t})\le x(\gamma_{+}^{t})$. On the other hand, since $I(\gamma^{s})\subset I(\tilde\gamma^{s})$,
there are only two possibilities at $s+1$:
either $x(\tilde\gamma^{s+1})> x(\gamma^{s+1})$ ($\tilde\gamma$ overtakes $\gamma$), which implies $x(\tilde\gamma^{t})>x(\gamma^{t})$ by Lemma~\ref{lemma:comparex};
or $I(\gamma^{s+1})\subset I(\tilde\gamma^{s+1})$.
In the latter case, if $s+1<t$, we then repeat the arguments used for $s$; otherwise, we have $I(\gamma^{t})\subset I(\tilde\gamma^{t})$. 
All cases considered,  we conclude that $\tilde\gamma\in \mathbf{L}^{+}_{t}$. Moreover, for such a $\tilde\gamma$, there are at most two times $s$ such that $(\tilde\gamma, s)\in \mathbf{M}_{t}(m)$:
the first time $s'$ when $I(\gamma^{s})\subset I(\tilde\gamma^{s})$ and $\tilde\gamma$ is the largest soliton containing $\gamma$ (i.e.~first step in overtaking), and the first time $s''$ after $s'$ when $\tilde\gamma$ moves forward (i.e.~second step in overtaking).
And for all $u\ge s''+1$, we have $x(\tilde\gamma^{u})> x(\gamma^{u})$ as implied by Lemmas~\ref{lem:xgamma} and~\ref{lemma:comparex}, so that $\gamma$ and $\tilde\gamma$ will never intersect again.
Since $s'$ and $s''$ do not necessarily belong to $[0, t)$, we only have $M_{t}(m)\le 2 \#\mathbf{L}_{t}^{+}$. For the other inequality, we note that if $\tilde\gamma\in \mathbf{L}^{-}_{t}$, then $x(\tilde\gamma)< x(\gamma_{-})$, which implies actually $x(\tilde\gamma)\le x(\gamma_{+})<x(\gamma)$.
As $x(\tilde\gamma^{t})> x(\gamma^{t})$, it follows that both $s' $ and $ s'' $ described previously belong to $[0, t)$. So the claimed inequalities follow.

Write whp to denote ``with high $ \hmu_k $-probability,'' and note that high $ \mu $ implies high $ \hmu_k $.
By~\eqref{eq:speedexists}, for every $ \varepsilon>0 $, $ x(\gamma^t) $ is in $ v_k t \pm \varepsilon t $ and $ x(\gamma_-^t) $ is in $ v_m t \pm \varepsilon t $ whp.
Using $ \theta $-invariance of $ \mu $, by the Ergodic Theorem applied to counting $ m $-solitons on a large interval, the number of $ m $-solitons of $ T^t \eta $ located in $ [v_k t \pm \varepsilon t , v_m t \pm \varepsilon t ] $ is within $ \orho_k(v_m-v_k)t \pm 3\varepsilon t $ whp.
Therefore, $ \# \mathbf{L}_t^- $ is in $ \orho_m (v_m-v_k)t \pm 3\varepsilon t $ whp.
Since $ \mathbf{L}_t^+ $ and $ \mathbf{L}_t^- $ can differ by at most two $ m $-solitons, this concludes the proof that $ \tfrac{1}{2t} M_t(m) \to \orho_m (v_m-v_k) $ in probability.

The proof for $N_{t}(\ell)$ is similar but simpler. The only difference is that for each $\tilde\gamma\in \Gamma_{\ell}\eta$, there is at most one time $s$ such that $(\tilde\gamma, s)$ is counted in the tally $N_{t}(\ell)$.
\end{proof}

To complete the proof of~\eqref{eq:hspeedseq}, it remains to drop the assumption that $\orho_m=0$ for all large $m$.
The above proof contains all the argument, except that combining~\eqref{eq:displacement} and~\eqref{eq:collisions} requires a limit and an infinite sum to commute. To prove the general case, we take expectation in~\eqref{eq:displacement} with $t=1$ and combine it with~\eqref{eq:vfrompalm}:
\begin{align}
v_k &= \int x(\gamma^1) \hmu_k (\dd \eta) =
\int \Big[ k - \sum_{m>k} k\,M_1(m) + \sum_{m<k} 2m\, N_1(m) \Big] \, \hmu_k(\dd \eta)
\\
& =
k - k\sum_{m > k} \int M_{1}(m) \hmu_{k}(\dd\eta)+ \sum_{m<k}2m \int N_1(m) \, \hmu_k(\dd \eta),
\end{align}
So it remains to show that $\int N_1(m) \, \hmu_k(\dd \eta) = \orho_m(v_k-v_m)$ and $\frac12\int M_1(m) \, \hmu_k(\dd \eta) = \orho_m(v_m-v_k)$.
Similarly to the argument in \S\ref{sub:speedexists}, we decompose
\[
N_{t+1}(m) = N^{\eta}_{1}(m) + N_1^{\hT_k \eta}(m) + N_1^{\hT_k^2 \eta}(m) + \dots + N_1^{\hT_k^t \eta}(m)
\]
with a slight abuse of notation.
By $\hT_k$-invariance, $\frac{1}{t}N_t^\eta(m)$ converges $\hmu_k$-a.s.\ to a random variable (i.e.\ a measurable function of $\eta$) whose average is $\int N_1^\eta(m) \, \hmu_k(\dd \eta)$.
On the other hand, by~\eqref{eq:collisions} this variable is constant and equal to $\orho_m|v_k-v_m|$.
The same argument works for $M_{1}(m)$.
This concludes the proof of~\eqref{eq:hspeedseq}, and also of Theorem~\ref{thm:hspeed}.

\subsection{Soliton speeds are positive and increasing}
\label{sub:speedpositive}

Here we show that $ v_k $ is positive and increasing in $ k $ without assuming that $ \orho_k>0 $.
Both follow from combining~\eqref{eq:hspeedseq} with the following bound on the total flow of solitons crossing the origin:
\begin{align}
\label{eq:maxflow}
2 \orho \cdot v
:=
2 \sum_m \orho_m v_m < 1
.
\end{align}
The idea is that $ \rho_m v_m $ is the average flow of solitons through the origin and each one takes $ 2 $ time steps. Before sketching the proof, we show that soliton speeds are positive and increasing.

Splitting and regrouping the sums in~\eqref{eq:hspeedseq} gives
\begin{align}
\label{ll2}
v_k
=
k - \sum_{m} 2 (k \wedge m) \orho_m v_m
+
v_k \, \sum_{m} 2 (k \wedge m) \orho_m
=
a v_k + b,
\end{align}
where $ a \leq 2\lambda <1 $ and $ b \geq k - 2k \orho \cdot v>0 $.
Hence, $ v_k>0 $.

Now, writing $ d_k := v_{k+1}-v_k > 0 $, by subtracting the above equation from itself we get
\begin{align*}
d_k
&=
1
+
d_k \sum_{m} 2 (k \wedge m) \orho_m
+
\sum_{m \geq k+1} 2 \orho_m (v_{k+1}-v_m)
=
a d_k + b + c,
\end{align*}
where
$ a \leq 2 \lambda < 1 $, $ b \geq 1 - 2\orho\cdot v > 0 $ and $ c = \sum_{m>k}2 \orho_m v_{k+1} \geq 0 $.
Therefore, $ d_k>0 $.

\begin{proof}
[Proof of~\eqref{eq:maxflow}]
The proof bears much similarity to that of~\eqref{eq:collisions}, so we only sketch the arguments here. Denote by $Q_{t}(k)=\#\{(\gamma, s): \gamma\in \Gamma_{k}\eta, 0\le s<t, 0\in I(\gamma^{s}) \text{ and } x(\gamma^{s+1})\ne x(\gamma^{s})\}$, the number of occurrences that a $k$-soliton is the largest one straddling $0$ between times $0$ and $t$.
Note that each $ k $-soliton that crosses the origin in the time interval $[0,t]$ contributes twice to $Q_{t}(k)$, 
once when the origin is in its tail and once in its head.

To find the growth rate of $Q_{t}(k)$ as $ t\to\infty $, let us denote by $\gamma_{+}$ the rightmost $k$-soliton $\gamma'$ satisfying $x(\gamma')\le 0$. The counts
$Q^{-}_{t}(k) = 2\#\{\gamma\in \Gamma_{k}\eta: 0<x(\gamma^{t})<x(\gamma_{+}^{t})\}$ and $Q^{+}_{t}(k) = 2\#\{\gamma\in \Gamma_{k}\eta: I(\gamma^{t})\cap [0, x(\gamma_{+}^{t})]\ne \varnothing\}$ satisfy $Q^{-}_{t}(k)-3 \le Q_{t}(k) \le Q^{+}_{t}(k)$. Arguments identical to those in~\eqref{eq:collisions} show that $\frac{1}{t}Q_{t}(k)\to 2\orho_{k}v_{k}$ in probability for each $k\ge 1$.

Applying with the Ergodic Theorem for $ T $, convergence of $\frac{1}{t}Q_{t}(k)$ says that $ 2 \orho_k v_k $ equals the probability that the largest soliton $ \gamma $ such that $ 0\in I(\gamma) $ belongs to $\Gamma_k \eta$.
For each $ \eta\in\cX $, either $ 0\in R\eta $ or there is a unique such $ k $, so $$ \mu(0\in R\eta) + \sum_k 2 \orho_k v_k = 1 ,$$ concluding the proof.
\end{proof}

\subsection{Recursion formulas for independent components}
\label{sub:speedexplicit}

We now prove Theorem~\ref{thm:speedsexplicit}.
Let $\zeta = M \eta$.
We are assuming that the field $\left( \zeta_k(i) \right)_{i\in\Z}$ is i.i.d.\ over $i$ for each $k$, and independent over $k$.
So let us proceed the other way around.
We let $P$ denote the law of $\zeta$ and $E$ the corresponding expectation.
In this notation, $\hmu$ is the law of $\eta=M^{-1} \zeta$.

Note that~\eqref{eq:skk1} and~\eqref{eq:rhok} give the first two equations in~\eqref{eq:hspeedexplicit}.
The third equation can be taken as the definition of $s_k$, and it is a simple recursive definition once one has $\rho$, $w$ and $\alpha$.
Combining these with~\eqref{eq:anothervfrompalm}, to get the last equation in~\eqref{eq:hspeedexplicit} we need to show that
\begin{equation}
\label{eq:totaldisplacement}
E \sum_{y=0}^{\ww(\eta)-1} \Delta^k_\eta(y)
=
\alpha_k \cdot s_k
.
\end{equation}
We now use the assumption that $\eta=M^{-1} \zeta$, where $M^{-1}$ denotes de reconstruction map of \S\ref{sub:reconstruction}.
First note that $$\sum_{y=0}^{\ww(\eta)-1} \Delta^k_\eta(y) = \sum_\gamma x(\gamma_1)-x(\gamma), $$
where the sum is over all $ k $-solitons $ \gamma $ located between Record~0 and Record~1.
Moreover, $k$-solitons appended to $k$-slots belonging to bigger solitons will stay put, just switching zeros for ones, and only the $k$-solitons which are appended directly to the $0$-th $k$-slot at $x=0$ will actually jump (Lemma~\ref{lem:xgamma}).
Furthermore, by Proposition~\ref{prop:solitontrack} and Lemma~\ref{lem:xgamma}, the size of their jump equals the distance between their leftmost $ 1 $ and their leftmost $ 0 $, that is, the distance between the tip of their head and the tip of their tail.

Now the number of $k$-solitons appended to the $0$-th $k$-slot is exactly $\zeta_k(0)$, which on average equals $\alpha_k$ by~\eqref{eq:alphak}.
So to conclude the proof of~\eqref{eq:totaldisplacement}, it is enough to observe that $s_k$ given by the recursion relation in~\eqref{eq:hspeedexplicit} in fact gives the expected distance between the tips of the head and tail of a typical $k$-soliton appended to a record.
By independence of $ \zeta_k $ over $ k $, being appended to a record is irrelevant, and we can consider a typical $ k $-soliton instead.

We claim that $ s_k $
equals the a.s.\ empirical average of the distance between the tip of the head and the tip of the tail among all $ k $-solitons,
and also that $ 2s_k $
equals the a.s.\ empirical average of the size of $ I(\gamma) $ among all $ k $-solitons $ \gamma $.
We made the previous statement stronger so we can prove it by induction.
Remember that the interval $ I(\gamma) $ consists of sites occupied by $ \gamma $ together with the smaller solitons appended to its slots.
For $k=1$ we have $s_1=1$, consistent with the fact that a $1$-soliton is always given by the strings $a=10$ or $\tilde{a}=01$ with nothing appended inside it.
For $k=2$, note that each $2$-soliton (including smaller solitons appended to its slots) is of the form $b=11\tilde{a}^*00{a}^*$ or $\tilde{b}=00{a}^*11\tilde{a}^*$ where $a^*$ stands for $\zeta_1(i)$ copies of $a$ and $\tilde{a}^*$ stands for $\zeta_1(j)$ copies of $\tilde{a}$ for some $i,j$ which are determined by $(\zeta_m)_{m\geq2}$.
So the average size of $11\tilde{a}^*$ and that of $00a^*$ both equal $2+2\alpha_1$.
For $k=3$, note that each $3$-soliton is of the form
${c}=11\tilde{a}^*1\tilde{a}^*\tilde{b}^*00{a}^*0{a}^*{b}^*$
or
$\tilde{c}=00{a}^*0{a}^*{b}^*11\tilde{a}^*1\tilde{a}^*\tilde{b}^*$
where $b^*$ stands for $\zeta_2(i)$ independent copies of $b$, etc.
So the average size of each half of a $3$-soliton equals $s_3=3 + 2s_2 \alpha_2 + 4 s_1 \alpha_1$.
The induction step is clear, which concludes the proof of Theorem~\ref{thm:speedsexplicit}.

\subsection{Vertical speed}
\label{sub:vertical}

A previous version of this work [\href{https://arxiv.org/abs/1806.02798v3}{arXiv:1806.02798v3}] focused on measuring vertical speed of solitons, before the authors realized that it was in fact possible to study the horizontal speeds.

We now briefly mention the results about vertical speed.
The analysis is very similar to what is now done in \S\ref{sub:speedequation} above, and tedious details will be omitted.

For $\beta \in R\eta$, we define
\(
\beta^{t}:=r(T^t \xi,j),
\text{ where }
\beta=r(\xi,j)
.
\)
Note that this definition does not depend on the lift $\xi[\eta]$.
Recalling~\eqref{eq:trackslot},
we define the displacement of a $k$-bearer $\pi\in S_k\eta$ measured in terms of records by
\[
y_k^t(\eta,\pi) = \# \left\{ \beta\in R\eta: \pi < \beta \text{ and } \pi^{k,t} \ge \beta^t \right\}
,
\quad
\pi \in S_k \eta
.
\]
In case there is a $k$-soliton $\gamma\in\Gamma_k\eta$ appended to the $k$-slot $\pi$ in $\eta$, the tagged $k$-soliton $\gamma^t$ will appear appended to the $k$-slot $\pi^{k,t}$ in $T^t\eta$, so $y_k^t$ also measures the displacement of tagged $k$-solitons, but it is well defined even when there are no $k$-solitons. The limit~\eqref{eq:6.2} below gives a physical meaning to the soliton speeds $v_k$ when $\rho_k=0$ in Theorem~\ref{thm:hspeed}.

\begin{theorem}
\label{thm:speeds}
Let $\mu$ be a measure on $\cX$ such that under $\hmu$ each $k$-th component $M_k \eta$ is i.i.d.\ and they are independent over $k$.
There exists a non-decreasing deterministic sequence $h=(h_k)_{k\ge 1}$ such that, $\mu$-a.s.\ on $\eta$, for all $k\in\N$ and $\pi\in S_k\eta$,
\begin{align}
\label{eq:vk0}
\lim_{t\to\infty} \frac{y_k^t(\eta,\pi)}{t} &= h_k
\in [k,\infty].
\end{align}
Assuming
\(
\sum_{k} k^2 \rho_k < \infty,
\)
the vector $(h_k)_{k\ge 1}$ is the unique finite solution of the linear system
\begin{align}
\label{eq:veleq}
h_k &= k + \sum_{m>k} 2(m-k)(h_m-h_k)\rho_m ,\quad k\ge 1
.
\end{align}
The asymptotic speed of tagged records is given by
\begin{align}
\label{eq:6.1}
\lim_{t\to\infty} - \frac{\beta^t}{t}
=
v_0
:=
\sum_{m\ge 1} 2m \rho_m h_m,
\quad
\text{for all } \beta\in R\eta,\
\mu\text{-a.s.}
\end{align}
The asymptotic speed $v_k$ of $k$-bearers is also given by
\begin{align}
\label{eq:6.2}
\lim_{t\to\infty} \frac{\pi^{k,t}}{t} =
v_k =
h_k w_0 - v_0
,\quad
\text{ for all } \pi\in S_k\eta,\
\mu\text{-a.s.}
,
\end{align}
From the two last equations, the speed of tagged records is also given by
\begin{align}
\label{eq:anothervzero}
v_0=\sum_{m\ge 1} 2m \rho_m v_m.
\end{align}
Furthermore, the vertical speed of any walk representation $\xi=\xi[\eta]$ is given by
\begin{equation}
\label{eq:h0}
\lim_{t\to\infty} -\frac{T^t\xi(0)}{t} = h_0 := \frac{v_0}{w_0}.
\end{equation}
\end{theorem}

Let us outline the argument for the first four equations.

The proof of~\eqref{eq:vk0} is similar to that of~\eqref{eq:speedexists} and only uses $ \theta $-ergodicity of $ \eta $.
Each time an $m$-soliton overtakes a $k$-bearer, this causes the position of the $k$-bearer measured in records to be incremented by an extra factor of $2(m-k)$. On the other hand, the position of a $ k $-bearer measured in records is not affected by overtaking smaller solitons.
These two facts explain the origin of~\eqref{eq:veleq}, but the proof uses a truncation argument that deletes all large components, and the i.i.d.\ assumption is to ensure that the resulting configuration is still $ \theta $-ergodic.

The speeds $ h_k $ are either all finite or all infinite, they are finite if and only if $ \sum_m m^2 \rho_m < \infty $.
Deleting large solitons, one gets from~\eqref{eq:veleq} that $ 0 \leq h_m - h_k \leq m-k $.
Plugging this back into~\eqref{eq:veleq} and using $ \sum_m m \rho_m < \infty $, one eventually gets $ h_{k+j} - h_k \geq \delta j $ for some $ \delta>0 $ and all large $ k $.
This gives $ \delta k \leq h_k \leq k + 2 \sum_m m^2 \rho_m $ uniformly with respect to the deletion threshold, which yields stated equivalence.
Also, under this condition, one can show existence and uniqueness of the solution to~\eqref{eq:veleq} using truncation, as we did for~\eqref{eq:skk1}.

Each time a tagged $m$-soliton crosses a tagged record from left to right, it causes the record to move $2m$ boxes left.
On the other hand, by mass conservation the number of such crossings by time $t$ equals $\rho_m y_m^t$ which is about $\rho_m h_m t$ by~\eqref{eq:vk0}.
Summing over $m$ we get~\eqref{eq:6.1}.
As a side remark, the vertical speed $h_0$ is given by the expected vertical jump at the origin $h_0 = \int (\xi[\eta](0)-T\xi[\eta](0)) \mu(\dd\eta)$, and intuitively this is related to $ \sum_k k^2 \rho_k $ because the excursion contained the origin is size-biased.

Finally, by~\eqref{eq:vk0}, the $k$-bearer $\pi = 0 \in S_k \xi$ will typically have crossed about $h_k t$ records by time $t$, so it will be between two tagged records with initial index about $h_k t$.
By ergodicity, the initial position of these records is about $w_0 h_k t$, so by~\eqref{eq:6.1} their position at time $t$ will be about
$w_0 h_k t - v_0 t$.
Dividing by $t$ and taking a limit one gets~\eqref{eq:6.2}.

From~\eqref{eq:6.2} we have $h_m = \frac{v_0+v_m}{w_0}$, substituting into~\eqref{eq:6.1} and using~\eqref{eq:wfinite} we get~\eqref{eq:anothervzero}.
To prove~\eqref{eq:h0} we note that after $t$ iterations of $T$, Record~$i$ will be at $x=o(t)$ if $r(\xi,i) = v_0 t + o(t)$, which implies that $T^t \xi(0) = -i + o(t)$.
On the other hand, $r(\xi,i) = w_0 i + o(i)$, whence $T^t \xi(0) = -h_0 t + o(t)$, concluding the proof.

\section{Postponed proofs}
\label{sec:postponed}

\begin{proof}
[Proof of Proposition~\ref{prop:solitontrack}]
Let us prove for finite $\eta$ first.
The proof is by induction on number of balls contained in $\eta$.
Identifying $0$ with ``$\ominus$'' and $1$ with ``$\oplus$'', consider the following data stream version of the TS-Algorithm.

\begin{algorithm}[H]
Start with the word $\ominus^\infty$ which is semi-infinite to the left \\
\For{\rm each symbol in the finite configuration $\eta$}{
Append the symbol to the word
\\
Perform annihilation if the two last runs have the same length
\\
Symbols that annihilate correspond to a soliton
}
\end{algorithm}

For example, for the finite sequence $\eta=\oplus\oplus\ominus\oplus\oplus\ominus\ominus\oplus\oplus\oplus\ominus\ominus\ominus\ominus\ominus$ the algorithm would produce the words
$\ominus^\infty\oplus$,
$\ominus^\infty\oplus^2$,
$\ominus^\infty\oplus^2\ominus$,
$\ominus^\infty\oplus^2\underline{\ominus\oplus}$,
$\ominus^\infty\oplus^3$,
$\ominus^\infty\oplus^3\ominus$,
$\ominus^\infty\oplus^3\ominus^2$,
$\ominus^\infty\oplus^3\ominus^2\oplus$,
$\ominus^\infty\oplus^3\underline{\ominus^2\oplus^2}$,
$\ominus^\infty\oplus^4$,
$\ominus^\infty\oplus^4\ominus$,
$\ominus^\infty\oplus^4\ominus^2$,
$\ominus^\infty\oplus^4\ominus^3$,
$\ominus^\infty\underline{\oplus^4\ominus^4}$,
and
$\ominus^\infty\ominus=\ominus^\infty$,
identifying a $1$-soliton, a $2$-soliton and a $4$-soliton.
For the example in Fig.~\ref{fig:algo}, it produces
$\ominus^\infty\oplus$,
$\ominus^\infty\oplus^2$,
$\ominus^\infty\oplus^3$,
$\ominus^\infty\oplus^4$,
$\ominus^\infty\oplus^4\ominus$,
$\ominus^\infty\oplus^4\ominus^2$,
$\ominus^\infty\oplus^4\ominus^2\oplus$,
$\ominus^\infty\oplus^4\ominus^2\underline{\oplus\ominus}$,
$\ominus^\infty\oplus^4\ominus^2\oplus$,
$\ominus^\infty\oplus^4\underline{\ominus^2\oplus^2}$,
$\ominus^\infty\oplus^5$,
$\ominus^\infty\oplus^5\ominus$,
$\ominus^\infty\oplus^5\underline{\ominus\oplus}$,
$\ominus^\infty\oplus^5\ominus$,
$\ominus^\infty\oplus^5\ominus^2$,
$\ominus^\infty\oplus^5\ominus^3$,
$\ominus^\infty\oplus^5\ominus^4$,
$\ominus^\infty\oplus^5\ominus^4\oplus$,
$\ominus^\infty\oplus^5\ominus^4\oplus^2$,
$\ominus^\infty\oplus^5\ominus^4\oplus^2\ominus$,
$\ominus^\infty\oplus^5\ominus^4\oplus^2\underline{\ominus\oplus}$,
$\ominus^\infty\oplus^5\ominus^4\oplus^3$,
$\ominus^\infty\oplus^5\ominus^4\oplus^3\ominus$,
$\ominus^\infty\oplus^5\ominus^4\oplus^3\ominus^2$,
$\ominus^\infty\oplus^5\ominus^4\underline{\oplus^3\ominus^3}$,
$\ominus^\infty\underline{\oplus^5\ominus^5}$,
identifying three $1$-solitons, a $2$-soliton, a $3$-soliton, and a $5$-soliton.

Let us call \emph{$\oplus$-alternating suffix} (or simply \emph{$\oplus$-suffix}) a finite word $\omega$ which is either empty or starts with $\oplus$ and such that each run in the word is strictly longer than the next one.
So the above algorithm always produces words given by $\ominus^\infty$ followed by a $\oplus$-suffix.
We define \emph{$\ominus$-suffix} in the obvious way.
The \emph{net value} $v(\omega)$ of a finite suffix $\omega$ is the number of $\oplus$'s minus the number of $\ominus$'s.
We will use the following observation about the above procedure.

Observation~1. The net value of a non-empty $\oplus$-suffix $\omega$ is positive and it is at most equal to the length $\ell_1(\omega)$ of its first run (e.g.\ for $\cdots\oplus^4\ominus^3\oplus$ we have $0 < 2 \le 4$).
In particular, $v(\omega)=\ell_1(\omega)$ only if it consists of a single run.

Observation~2. The net value of a finite suffix $\omega$ equals the net value of the portion of $\eta$ that generated it, which in turn is given by the net increase in $\xi[\eta]$.

Observation~3. If the suffixes $\omega_1,\dots,\omega_n$ produced while processing a certain piece of $\eta$ are all $\oplus$-suffixes, then $\ell_1(\omega_n)$ equals the maximal net value of $\omega_i$ for $i=1,\dots,n$.
In particular, if $v(\omega_n)=\max_{i} v(\omega_i)$, then $\ell_1(\omega_n)=v(\omega_n)$ and, by Observation~1, $\omega_n$ consists of a single run.

To prove the proposition we will split a finite $\eta$ into three blocks and analyze how they interact under the data stream algorithm, both before and after the application of $T$, as shown in Fig.~\ref{fig:proof-prop-conserv}.

\begin{figure}[b]
\centering
\includegraphics[width=.8\textwidth]{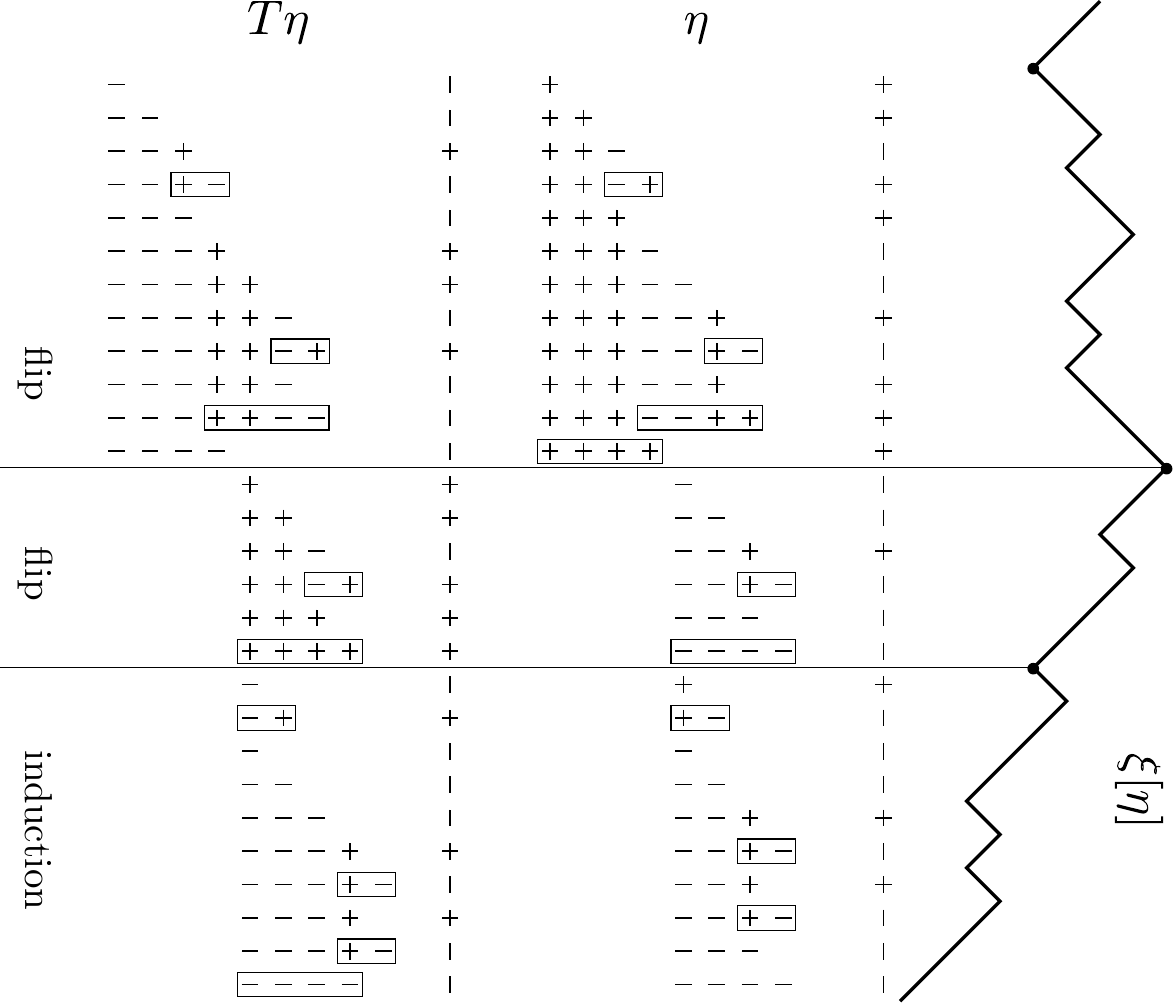}
\caption{\small%
Example showing conservation of solitons by splitting space in three parts: rising, falling and remainder. After applying $T$, the configuration on the rising and falling parts are flipped, the smaller solitons are conserved and flipped, the biggest soliton moves forward, and will have its tail in the remainder part. Applying $T$ to the remainder part conserves solitons by induction.}
\label{fig:proof-prop-conserv}
\end{figure}

Define the first non-empty soft excursion as the piece of $\eta$ going from the first $\oplus$ until the first point that makes the net value equal zero.
Split this excursion into \emph{rising} and \emph{falling} parts as follows.
The rising part goes until the point where the net value $k$ is maximal (in case the maximum is attained more than once, take the rightmost one), and the falling part consist of the remaining boxes, until the end of the first soft excursion.
The \emph{remainder} consists of all the sites to the right of the falling block.
Let $I_1,I_2,I_3 \subseteq \Z$ denote these sets of sites.

By definition of $I_1$ and by Observation~2, the streaming algorithm applied to $\eta$ on $I_1$ always produces a non-empty $\oplus$-suffix, its net value is always at most $k$ and ends being equal to $k$.

By Observations~1 and~3, the word produced by the algorithm after processing this first block is $\oplus^k$.
By similar considerations, the algorithm applied to $\eta$ on $I_2$ always produces non-empty $\ominus$-suffixes whose net values are strictly between $-k$ and $0$, except for the final step when it produces $\ominus^k$.

Hence, when processing $\eta$ on $I_1 \cup I_2$, the $\oplus^k$ obtained after processing the rising part is kept untouched until the very end, when it is annihilated by the $\ominus^k$ obtained after processing the falling part.
So when the algorithm starts processing $\eta$ on $I_3$ there is no suffix left by the previous steps and this part of $\eta$ is decomposed into solitons just as it would if it was processing $\eta_{|_{I_3}}$ instead.

Now notice that, by the definition of $T$ on $\xi[\eta]$, the net value of $T \eta$ on any prefix of $I_3$ is non-positive.
Indeed, at the rightmost site $y$ of $I_2$, the walk $\xi$ coincides with its running minimum, so $T\xi(y)=\xi(y)$ and $T\xi(x) \le T\xi(y)$ for all $x>y$.
Hence, applying the streaming algorithm to this portion of $T\eta$ produces a $\ominus$-suffix at all steps.

Also, since $\xi(x) \ge \xi(y)$ for all $x\in I_1 \cup I_2$, by definition of $T$ we have that $\eta$ and $T \eta$ are the complementary of each other on these two blocks.
So by the previous observations, the streaming algorithm applied to $\eta$ and to $T \eta$ on $I_1$ will produce exactly the opposite suffixes at every step.
The same is true for $I_2$.
The only difference is that now the $\ominus^k$ produced after processing $T\eta$ on $I_1$ is incorporated into the infinite prefix $\ominus^\infty$, and it will not annihilate with the $\oplus^k$ obtained after processing $T\eta$ on $I_2$.
Hence, while processing $T\eta$ on $I_1 \cup I_2$, the same solitons will be generated, with $\oplus$ replaced by $\ominus$, that is, with the head occupying the former position of the tail, except for this last $k$-soliton.

Finally, the $\oplus^k$ obtained after processing $T\eta$ on $I_1 \cup I_2$ will not increase its length while processing $T\eta$ on $I_3$, because processing $T\eta$ on $I_3$ always produces $\ominus$-suffixes.
So this run $\oplus^k$ is preserved until the first time when the processing of $T\eta$ on $I_3$ produces a $\ominus^k$, and they both annihilate.
This eventually occurs because $T \eta$ has infinitely many records to the right.
So again the head of the corresponding $k$-soliton will take the position previously occupied by the tail of a $k$-soliton.
Moreover, when it occurs, it annihilates $\ominus$'s that were not going to be annihilated while processing $(T\eta)_{|_{I_3}}$
because they would have been simply absorbed by the prefix $\ominus^\infty$.
Hence, the presence of this $\oplus^k$ does not change how the algorithm processes $T\eta$ on $I_3$, neither before nor after such annihilation occurs.
To conclude, note that $\eta_{|_{I_3}}$ contains fewer balls than $\eta$ so we can assume by induction that the tails of all the solitons of $\eta_{|_{I_3}}$ will become the heads of the solitons of $T\eta_{|_{I_3}}$, proving the proposition for the case of a finite configuration $\eta$.

We finally consider general $\eta\in\cX$.
Let $A$ be a set of $k$ sites.
Let $y_2,y_3$ be records for $T\eta$ to the left and right of $A$, respectively.
Let $y_1<y_2$ and $y_4>y_3$ be records for $\eta$.
Let $\eta'$ denote the restricted configuration, given by $\eta'(x):=\eta(x)\one\{x\in[y_1, y_4]\}$.
Since solitons are always contained in the interval between two consecutive records, if some $\gamma\in\Gamma_k\eta$ intersects $A$ then it is contained in $[y_1,y_4]$. Since $\eta'\le\eta$, and $x$ being a record for $\eta$ is a non-increasing property of $\eta$, $y_1$ and $y_4$ are also records for $\eta'$.
Hence, the soliton configuration $\Gamma_k\eta$ restricted to $[y_1,y_4]$ coincides with $\Gamma_k\eta'$.
Now notice that $T\eta'=T\eta$ on $[y_1,y_4]$ and $T\eta'=0$ on $(-\infty,y_1]$.
In particular, $T\eta'=T\eta$ on $[y_2,y_3]$, $T\eta' \le T\eta$ on $(-\infty,y_2]$,
and thus $y_2,y_3$ are also records for $T\eta'$.
Hence, by the same argument as above, if some $\gamma\in\Gamma_k T\eta'$ intersects $A$ then it is contained in $[y_2,y_3]$, moreover $\Gamma_k T \eta$ restricted to $[y_2,y_3]$ coincides with $\Gamma_k T \eta'$ restricted to $[y_2,y_3]$.
Since $\eta'$ is a finite configuration, by the previous case this concludes the proof.
\end{proof}

\begin{proof}[Proof of Lemma~\ref{lem:basics}]
Let us first take a configuration with a finite number of balls and let $\gamma$ be a soliton.
The second statements follow from the definition of $I(\gamma)$. For the rest,
note that while implementing the Takahashi--Satsuma algorithm, when $\gamma$ is identified, it appears as a set of consecutive sites (after the previous steps which remove solitons lodged inside $I(\gamma)$). Moreover, the rules of the algorithm imply that solitons intersecting $I(\gamma)$ which are removed from the previous steps must have a strictly smaller size than $\gamma$ and are completely contained in $I(\gamma)$. A formal proof can be done by an induction on the number of solitons.
These properties can be extended to infinite configurations similarly as in the proof of Proposition~\ref{prop:solitontrack} above, since the algorithm applies to each excursion of the
configuration.
\end{proof}

\begin{proof}[Proof of Lemma~\ref{lem:xgamma}]
Assume in the first place there is a finite number of balls and we proceed by an induction on the number of solitons. The basic case when there is only one soliton is clear. Suppose that the statement holds true for a system with up to $n$ solitons and suppose that $\eta$ is a configuration with $n+1$ solitons. Let $k$ be the minimal size of its solitons and let $\gamma$ be the leftmost $k$-soliton. Denote by $a=x(\gamma)$.
Then Lemma~\ref{lem:basics} implies that $\cH(\gamma)\cup \T(\gamma)=[a, a+2k)$.

Let us first show $x(\gamma^{1})\ge x(\gamma)=a$.
If $x(\gamma^{1})<a$, combined with the fact $\cH(\gamma^{1})=\T(\gamma)$, this would imply $\cH(\gamma^{1})\cup\T(\gamma^{1})=[a-k, a+k)$, as $\gamma^{1}$ has minimal size in $T\eta$.
Moreover, in this case we would have $\T(\gamma^{1})=\{a-k, \dots, a-1\}$ and $\T(\gamma)=\{a, \dots, a+k-1\}$. If $a-1$ is a record of $\eta$, then $a$ will also be a record of $\eta$ as $\eta(a)=0$, which is absurd. If $a-1\notin R\eta$ but $a-2\in R\eta$, we readily check $a+1\in R\eta$ in this case, which is also impossible unless $k=1$. Proceeding in this way, we deduce that all the sites of $\T(\gamma^{1})$ must belong to the head of some soliton $\gamma'$, which must have size $>k$ as a result of our choice of $\gamma$. Note the site $a-k-1$ also belongs to the head of $\gamma'$; otherwise the TS algorithm would produce a $k$-soliton from the sites $[a-k, a+k]$. Since none of the sites in $[a-k-1, a+k]$ are records, they are flipped in $T\eta$, so that we have a run of at least $k+1$ 0's preceding $a$. But this contradicts with the rules of the TS algorithm to have a $k$-soliton on $[a-k, a+k]$. This means it is impossible to have $x(\gamma^{1})<a$.

If $a-1\in R\eta$, then we must have $\cH(\gamma)=[a, a+k)$ and $\T(\gamma)=[a+k, a+2k)$; otherwise we would have $k+1$ 0's followed by $k$ 1's, contradicting the TS algorithm. Then in $T\eta$, the sites in $[a, a+2k)$ get flipped while $a-1$ still has value 0. In particular, we cannot have $\gamma^{1}$ on the same sites as $\gamma$. But $x(\gamma^{1})\ge x(\gamma)$ and $\cH(\gamma^{1})=\T(\gamma)$. So the only possibility is $\T(\gamma^{1})=[a+2k, a+3k)$.

If instead $a-1\notin R\eta$, then it also gets flipped in $T\eta$. Note that we must have $\eta(a-1), \eta(a)$ having different values; otherwise the TS algorithm would not produce a soliton starting from $a$. In that case it is straightforward to check that $\gamma^{1}$ is the soliton on the sites $[a, a+2k)$ and then $\T(\gamma^{1})=\cH(\gamma)$.

So far we have shown $x(\gamma^{1})\ge x(\gamma)$ and that $\T(\gamma^{1})=\cH(\gamma)$ if and only if $a-1$, the $k$-slot to which $\gamma$ is appended is not a record; so it must belong to some soliton $\gamma'$ with size $>k$. Moreover, in this case, we have $I(\gamma)\cap I(\gamma')\ne\varnothing$, which yields $I(\gamma)\subset I(\gamma')$ by Lemma~\ref{lem:basics}. 
Applying the induction hypothesis to the configuration $\eta'$ obtained from $\eta$ by removing $\gamma$, we see that the statement also holds for the $n$ solitons of $\eta'$. Together  with the previous arguments for $\gamma$, this implies the statement holds for all the $n+1$ solitons of $\eta$, since if we have $I(\gamma_{1})\subset I(\gamma_{2})$, then inserting smaller solitons will not affect this.
To extend to the infinite configuration, we follow the same strategy as employed in the proof of Theorem 3.1, noting that $T\eta(x)$ depends only on $\eta|_{[r(x), x]}$ with $r(x)$ being the rightmost record preceding $x$; we omit the details here.
\end{proof}

\begin{proof}
[Proof of Lemma~\ref{lemma:comparex}]
By Lemma~\ref{lem:basics} we have $\max I(\gamma) < x(\tilde\gamma)$.
By Proposition~\ref{prop:solitontrack}, $ x(\gamma^1) \leq \max I(\gamma) $.
By Lemma~\ref{lem:xgamma}, $ x(\tilde\gamma) \leq x(\tilde\gamma^{1}) $.
Combining these we have $ x(\gamma^{1}) < x(\tilde\gamma^{1}) $, and by induction on $ t $ we get $ x(\gamma^{t}) < x(\tilde\gamma^{t}) $.
\end{proof}

\begin{proof}[Proof of~\eqref{eq:collision}]
By Lemma~\ref{lem:xgamma}, we have $x(\gamma^{1})=x(\gamma)$ if and only if there exists some $\hat\gamma$ satisfying $I(\gamma)\subset I(\hat\gamma)$. Moreover, Lemma~\ref{lem:basics} implies that if there is another soliton $\gamma'$ with $I(\gamma)\subset I(\gamma')$, then $\hat\gamma$ and $\gamma'$ are of different sizes and the smaller one is nested inside the larger one. In that case, we take $\tilde\gamma$ to be the soliton with maximal size satisfying $I(\tilde\gamma)\supseteq I(\gamma)$. Note that $\tilde\gamma$ is necessarily appended to a record (if not, we can find an even larger soliton which contains $\gamma$), and therefore $x(\tilde\gamma^{1}) \ne x(\tilde\gamma)$.

It remains to show that if $x(\gamma^{1})>x(\gamma)$, then their difference is given by $\sum_{\ell<k} 2\ell N^{\gamma}_{0}(\ell)$. According to Lemma~\ref{lem:xgamma}, we must have $\cH(\gamma)$ to the left of $\T(\gamma)$ in this case.
It then follows from Proposition~\ref{prop:solitontrack} that $x(\gamma^{1})$ is the leftmost site of $\T(\gamma)$, denoted as $t(\gamma)$.
To conclude, it suffices to note that $t(\gamma)-x(\gamma)=k+\sum_{\gamma'} 2\times\text{size of }\gamma'$, where the sum is over all the solitons $\gamma'$ lodged inside the left side of $\gamma$. One readily check this leads to the desired identity.
\end{proof}

\begin{proof}
[Proof of Proposition~\ref{prop:uniqueness}]
Let $c_k := \textstyle{\sum_{m\neq k} 2(m\wedge k)\orho_m }$.
Since $ \ell \orho_\ell \to 0 $, by dominated convergence we have 
$ c:=\sup_k c_k = \lim_k c_k = \sum_\ell 2\ell\orho_\ell < \frac{1}{2} $ using our assumption. Define
\begin{align}
q_{k,m} :=
\begin{cases}
\frac{1}{c} 2(m\wedge k)\orho_m &\text{if } m\neq k\\
\frac{1}{c} (c-c_k) &\text{if } m=k
\end{cases}
\qquad
\text{ and }
\qquad
Q:=(q_{k,m})_{k,m}
.
\end{align}
Since $ \sum_m 2m \orho_m<\infty $, we can split the sums and rewrite~\eqref{eq:hspeedseq} as
\begin{align}
\label{eq:vkpp}
v_k &= k + c\, v_k - c\, (Qv)_k\,=\, \displaystyle{\frac{k -c\, (Qv)_k}{1-c} },
\quad
v_k \in [0,\infty)
.
\end{align}
Note that $(Qv)_k$ is finite for solutions $v$ to~\eqref{eq:vkpp}.

Let $v$ and $\tv$ be two solutions to~\eqref{eq:hspeedseq}, and denote $ r_k := v_k-\tv_k$ and $ s_k=|r_k| $.
Then
\begin{align}
s_k
=
|r_k| =
a \, |(Qr)_k|
\le
a \, (Qs)_k
,
\end{align}
where $a := c/(1-c) < 1$.
By induction, for all $ n\in\N $,
\begin{align}
\label{fkan}
s_k
\le
a^n \, (Q^n s)_k
.
\end{align}
The measure $\pi=(\pi_k)_{k\ge1}$ given by
$$\pi_k := \frac{\orho_k}{\sum_\ell \orho_\ell} $$
satisfies $\pi Q=\pi$.
Since $ \frac{1}{1-c}<2 $ and $ q_{k,m}\geq 0 $, solutions to~\eqref{eq:vkpp} satisfy $v_k\le 2k $ and
\begin{align}
\pi v
= \textstyle{\sum_k v_k\, \orho_k\, \bigl(\sum_\ell \orho_\ell \bigr)^{-1}\,
\le\, 2\sum_k k\,\orho_k\, \bigl(\sum_\ell \orho_\ell \bigr)^{-1}}
<\infty.
\end{align}
Likewise, $ \pi \tilde{v}<\infty $ and $ \pi s \leq \pi v + \pi \tilde{v} < \infty $.
Using~\eqref{fkan} and $ Q $-invariance of $\pi $,
\begin{align}
\textstyle{
\pi_k s_k
\le
\pi s
\leq
a^n \, \pi Q^n s
=
a^n \, \pi s,
}
\end{align}
for every $ n \in\N $.
Hence, if $\orho_k>0$, then $\pi_k>0$ and $|r_k|=0$, implying $v_k=\tv_k$.
When $\orho_k=0$ we have $v_k$ is a function of $(v_m:\orho_m>0)$, implying uniqueness also in this case.
\end{proof}

\section*{Acknowledgments}

We thank Roberto Fernández for very helpful discussions.

{
\small
\setstretch{1}
\bibliographystyle{bib/leo2020alphasc}
\bibliography{bib/leo,bib/leo2,bib/refs-pablo}
}

\vfill

Pablo A. Ferrari$^1$, Chi Nguyen$^1$, Leonardo T. Rolla$^{12}$, Minmin Wang$^1$

$^1$ Argentinian National Research Council at the University of Buenos Aires
\\
$^2$ Corresponding Author. e-mail: leorolla@dm.uba.ar

\end{document}